\date{}
\documentclass[11pt]{article}
\usepackage{amsmath,mathrsfs}
\usepackage{tikz}
\usepackage{tikz-3dplot}
\usetikzlibrary{calc,arrows.meta}
\usepackage{pgfplots}
\usepackage{soul}

\tdplotsetmaincoords{72}{118}
\usepackage{orcidlink}
\usepackage{authblk}
\usepackage{hyperref}
\usepackage{tikz}
\usetikzlibrary{snakes}
\usetikzlibrary{decorations.text,calc,arrows.meta}
\usepackage[vcentermath]{youngtab}
\usepackage{tkz-euclide}
\usepackage{pgfplots}
\usetikzlibrary{decorations.shapes}
\usetikzlibrary{positioning}
\usetikzlibrary{patterns}
\usetikzlibrary{shapes.geometric}
\usetikzlibrary{decorations.pathmorphing}
\usetikzlibrary{arrows.meta}
\tikzset{snake it/.style={decorate, decoration=snake}}
\usetikzlibrary{arrows,shapes,positioning}
\usetikzlibrary{decorations.markings}
\tikzstyle arrowstyle=[scale=1]
\tikzstyle directed=[postaction={decorate,decoration={markings,mark=at position .65 with {\arrow[arrowstyle]{stealth}}}}]
\tikzstyle reverse directed=[postaction={decorate,decoration={markings,mark=at position .65 with {\arrowreversed[arrowstyle]{stealth};}}}]

\usepackage{amsmath,amsthm,amsfonts,amssymb,amsopn,amscd} 
\usepackage{color}

\def\tilde{\widetilde}

\def\sI{{\mathscr I}}

\def\om{\omega}
\def\Om{\Omega}

\newcommand{\bthm}{\begin{theorem}}
\newcommand{\ethm}{\end{theorem}}

\newcommand{\bprop}{\begin{proposition}}
\newcommand{\eprop}{\end{proposition}}
\newcommand{\bcor}{\begin{corollary}}
\newcommand{\ecor}{\end{corollary}}
\newcommand{\blem}{\begin{lemma}}
\newcommand{\elem}{\end{lemma}}
\newcommand{\econj}{\end{conjecture}}

\def\setminus{\smallsetminus}

\def\RR{{\mathbb R}}
\def\A{{\cal A}}

\def\D{{\cal D}}

\def\M{{\cal M}}
\def\N{{\cal N}}
\def\R{{\cal R}}

\def\H{{\cal H}}

\def\S{{\cal S}}

\def\f{{\varphi}}

\def\PSL{{{\rm PSL}(2,\mathbb R)}}

\def\h{{\frak h}}

\def\S2{S^{1(2)}}

\def\RR{\mathbb R}
\newtheorem{theorem}{Theorem}[section]
\newtheorem{lemma}[theorem]{Lemma}
\newtheorem{conjecture}[theorem]{Conjecture}
\newtheorem{corollary}[theorem]{Corollary}

\newtheorem{proposition}[theorem]{Proposition}

\theoremstyle{definition} 

\theoremstyle{remark} \newtheorem{remark}[theorem]{Remark}

\newcommand{\ben}{\begin{equation}}
\newcommand{\een}{\end{equation}}

\def\setminus{\smallsetminus}

\def\PSL{PSU(1,1)}

\def\SL2{{{\rm SL}(2,\R)}}

\def\PSL2{{{\rm PSL}(2,\Reali)}}

\def\U1{{{\rm V}(1)}}
\def\SU2{{{\rm SV}(2)}}

\def\SU{{{\rm SU}}}

\def\A{{\mathcal A}}

\def\D{{\mathcal D}}

\def\H{{\mathcal H}}

\def\M{{\mathcal M}}
\def\N{{\mathcal N}}

\begin{document}

\author{Stefan Hollands \orcidlink{0000-0001-6627-2808}}
\affil{\small Institute for Theoretical Physics, Leipzig University, Br\" uderstrasse 16, 04103 Leipzig, and MPI-MiS, Inselstrasse 22, 04103, Leipzig, Germany, stefan.hollands@uni-leipzig.de}

\author[2]{Roberto Longo}
\author[2]{Gerardo Morsella}
	\affil[]{\small
	Dipartimento di Matematica, Università di Roma Tor Vergata
via della Ricerca Scientifica 1, I-00133 Roma, Italy
}%

\title{\Huge{A Quantum Dominant Energy Condition}}

\maketitle

\begin{abstract}
We propose a quantum dominant energy condition (QDEC) for the stress tensor in the context of quantum field theory in curved spacetimes. 
A rigorous proof is given for the case of Rindler wedges, and a heuristic discussion about possible generalizations to 
more general geometric setups, including curved spacetime, is provided. In Minkowski spacetime, we establish a connection with state recovery bounds. We 
illustrate the QDEC for coherent states of the free scalar field, where it turns out to be related to the ordinary DEC for the stress 
tensor of the classical solution that defines the coherent state.
\end{abstract}

\section{Introduction}

The dominant energy condition (DEC) states that $T_{\mu\nu} k^\mu u^\nu \ge 0$ for any future pointing causal (i.e., timelike\footnote{Equivalently, one may restrict to 
a pair of future pointing null vectors $u$ and $k$.} or null) vectors $u=u^\mu \partial_\mu$ and $k=k^\mu \partial_\mu$. 
It implies the null energy condition, $T_{\mu\nu}k^\mu k^\nu \ge 0$ for all null $k$, as a special case, and 
may be interpreted as saying that the energy-momentum flux as seen by an arbitrary observer is causal and future-pointing.
The DEC is used in several important theorems
in general relativity, such as the positive mass theorem, the Hawking topology theorem, various incompleteness theorems, or the conservation theorem, see e.g., \cite{wald} 
as a general reference. For instance, the conservation theorem \cite{HE} states that if the DEC holds in a, say globally hyperbolic, subregion $O$ of a spacetime and if $T_{\mu\nu}$ vanishes near a Cauchy surface of $O$, then it must vanish within all of $O$. In this way, the DEC is related to causality.

In quantum field theory (QFT), on can find for a given spacetime point a sequence of quantum states such that the DEC fails arbitrarily badly at that point \cite{EG}.
However, more non-local conditions may well hold. For instance, in Minkowski spacetime, if $k$ is a constant timelike vector, then the integral of $F^\nu=-T_{\mu}{}^\nu u^\mu$ over a spacelike hyperplane orthogonal to the future directed time-like vector $u$ is the generator of spacetime translations which has its spectrum in the closed future lightcone any QFT. 
Thus, its expectation value is a future pointing causal vector. 
A more interesting, less non-local relationship, connecting the {\it null} components of the stress energy to the entanglement entropy of a region called the QNEC has been proposed by \cite{B1,B2}. Here we propose such a condition for the DEC. 

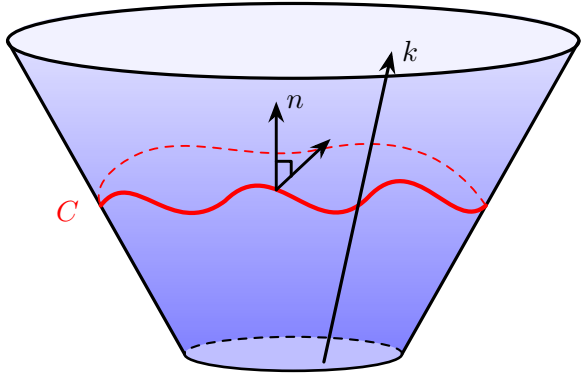
\begin{figure}[t]
\centering

\begin{tikzpicture}[
    scale=0.9,
    line cap=round,
    line join=round,
    >=Stealth
]
\def\Rt{4.2}
\def\Rb{1.6}
\path[
    left color=blue!8,
    right color=blue!50,
    shading angle=0,
    draw=none
]

(-4.2,4.6)

.. controls (-4.2,4.90)
           and (-3.0,5.15) ..
(0,5.15)

.. controls (3.0,5.15)
           and (4.2,4.90) ..
(4.2,4.6)

-- (1.6,0)

.. controls (1.0,-0.10)
           and (-1.0,-0.10) ..
(-1.6,0)

-- cycle;
\fill[blue!5]
(0,4.6)
ellipse[x radius=\Rt,y radius=0.55];

\draw[line width=1.2pt]
(0,4.6)
ellipse[x radius=\Rt,y radius=0.55];

\fill[blue!25]
(0,0)
ellipse[x radius=\Rb,y radius=0.25];

\draw[line width=1pt]
(-\Rb,0)
arc[
    start angle=180,
    end angle=360,
    x radius=\Rb,
    y radius=0.25
];

\draw[dashed,line width=.8pt]
(\Rb,0)
arc[
    start angle=0,
    end angle=180,
    x radius=\Rb,
    y radius=0.25
];
\draw[line width=1.2pt]
(-\Rt,4.6)--(-\Rb,0);

\draw[line width=1.2pt]
(\Rt,4.6)--(\Rb,0);
\coordinate (L) at (-2.83,2.18);
\coordinate (R) at ( 2.83,2.18);

\draw[line width=1.6pt,color=red]

(L)

.. controls (-2.35,2.80)
           and (-1.80,1.65) ..
(-1.00,2.25)

.. controls (-0.25,3.00)
           and (0.35,1.55) ..
(1.10,2.30)

.. controls (1.80,3.05)
           and (2.35,1.75) ..
(R);

\draw[dashed,line width=.7pt,color=red]

(R)

.. controls (2.20,3.10)
           and (1.30,3.20) ..
(0.35,3.00)

.. controls (-0.55,2.80)
           and (-1.50,3.25) ..
(-2.25,2.95)

.. controls (-2.70,2.75)
           and (-2.95,2.45) ..
(L);

\node[left] at (-3,2.1) {{\color{red} $C$}};

\coordinate (P) at (-0.25,2.41);

\draw[
    very thick,
    ->,
    line width=1.1pt
]
(P)--++(0,1.3);

\draw[
    very thick,
    ->,
    line width=1.1pt
]
(P)--++(0.8,0.75);

\draw[line width=1.1pt]
($(P)+(0,0.42)$)
--
($(P)+(0.22,0.42)$)
--
($(P)+(0.22,0.20)$);

\node[right]
at ($(P)+(0.0,1.3)$)
{$n$};

\coordinate (K0) at (0.45,-0.12);

\draw[
    very thick,
    ->,
    line width=1.2pt
]
(K0)--(1.45,4.45);

\node[right]
at (1.45,4.45)
{$k$};

\end{tikzpicture}

\caption{
 \cite{ChatGPT}
The cut $C$ with outgoing null sheet tangent to the vector $k=k^\mu \partial_\mu$.
$n=n^{\mu\nu}\partial_\mu \wedge \partial_\nu$ is the bi-normal attached to the cut.
}

\label{fig:QDEC}

\end{figure}

Our quantum DEC (QDEC) is an information-energy tradeoff relation. 
It states that if $C$ is an entangling cut, i.e., the \textcolor{black}{boundary} of a sufficiently regular domain within a Cauchy slice, if $k$ is the up to normalization unique outgoing, future directed null 
vector field orthogonal to $C$, then 
\ben
\label{qdec}
\left( \langle T^{\mu\nu}\rangle - \frac{\hbar}{2\pi} \, n^{\mu\alpha} n^{\nu\beta} \, \nabla_\alpha \frac{\delta}{\delta C^\beta} S_{\rm EE} \right) k_\nu u_\mu \ge 0
\een
for any future directed causal vector $u$ orthogonal to $C$, and 
for any sufficiently regular quantum state $\Phi$. Here, $\langle X \rangle = (\Phi, X\Phi)$ is the expectation value of an operator, 
$S_{\rm EE}[C]$ is the entanglement entropy of the domain enclosed by the cut $C$ in the state $\Phi$. $n=n^{\mu\nu}\partial_\mu \wedge \partial_\nu$ is the bi-normal\footnote{Our normalizations are 
such that $n^{\mu\nu} n_{\mu\nu} = -2$.} to the cut\footnote{I.e., a smooth, codimension one submanifold of the null-sheet intersecting each null geodesic generator precisely once.} $C$, see fig. \ref{fig:QDEC}.  The symbol $\delta/\delta C$ indicates 
the variational derivative with respect to changes in the location of $C$, see \eqref{vardef}. 

We give evidence for the validity of the QDEC in the case of a half space on a time slice in Minkowski spacetime, such that $k$ is tangent to the future Rindler horizon in the causal wedge forming the causal envelope of this half-space. In the case of $d=2$ spacetime dimensions, we state the QDEC as a certain convexity condition of the 
relative entropy $S(\Phi |  \! | \Omega)$ between $\Phi$ and the vacuum state $\Omega$ of the QFT, when viewed as a function of the entangling cut, i.e. simply the point where 
the future and past Rindler horizons meet. In the special case of coherent states of a free real Klein-Gordon field, we show that this form of the QDEC simply follows from classical DEC of the stress energy tensor of the classical solution underlying the coherent state, see Sec. \ref{coherent}.

Our formulation of the QDEC in terms of the relative entropy avoids the explicit mentioning of the entanglement entropy---which undifferentiated is a formally infinite quantity in QFT---and thereby sets a reasonable setting for a rigorous proof for a wide class of states that is, e.g., dense in the Hilbert space. 
Our proof is based on a variational principle \cite{W1} that has been adapted in rigorous proofs of the quantum null energy condition (QNEC) \cite{CF18,HL25b}. These proofs, as well as those in the present paper, use the theory of half-sided modular inclusions \cite{B1,B2,Wies1,AZ05,Florig}, which gives an operator that formally equals a suitable null-component of the stress energy tensor smeared over a horizon \cite{Markov}. 

This formal relationship leads to the version of \eqref{qdec} in $d=2$. In that case, $C$ is just a point, the variational derivative is just the ordinary derivative with respect to this point, and the corresponding relation was already conjectured by \cite[Eq. (39)]{W1}. By making a particular choice of the 
vectors, we have, e.g., $\langle T_{00} \rangle \ge (\hbar/2\pi) \partial_1^2 S_{\rm EE}$ as a special case in $d=2$.
In the case of $d>2$ dimensions, our result is not stated as a simple convexity condition for the relative entropy, but, as we outline, 
it still heuristically gives \eqref{qdec} in the case of Rindler horizons. 

At present, it is unclear to us whether the QDEC \eqref{qdec} should be expected for cuts more general than the edge of a Rindler wedge in Minkowski spacetime, or even curved spacetimes. We postpone a discussion of some of the issues that arise to Sec. \ref{extensions} but remark that there will most likely be interesting instances where the condition applies, such as idealized wormholes, or degenerate Killing horizons which are present e.g., in the extremal Kerr-Newman-(A)dS black hole spacetimes. However, regardless of these considerations, already the QDEC for Rindler wedges in Minkowski spacetime is interesting, and implies, in particular, a non-trivial inequality between the expected energy of a quantum state on a Rindler wedge, and the amount of recoverable information of its restriction to a sub-wedge, i.e., after some radiation may have leaked out, see Sec. \ref{staterecovery}.

\medskip
\noindent
{\bf Note added in proof:} While we were working on this project, we were notified by Aron Wall and Zihan Yan that they have been considering a QDEC independently. We thank them for correspondence.  

\section{Notations and preliminaries}

\subsection{Relative entropy and relative modular theory}
\label{modthbas} 

In this paper, von Neumann algebras are denoted by caligraphic letters such as $\M$, and are understood to act on a on a Hilbert space $\H$. 
$\M'$ denotes the commutant, i.e., the von Neumann algebra of all bounded operators on $\H$ commuting with every operator from $\M$.
If we have two cyclic and separating  vectors $\Om,\Phi\in\H$ for $\M$, then following \cite{Ar76,Ar77}, 
we define the relative Tomita's operators on $\H$ as the closures of
\[
S_{\Om,\Phi; \M} : m\Phi \mapsto m^*\Om\, , \quad m\in \M\, ,
\]
\ben\label{SF}
S_{\Om,\Phi;\M'}: m'\Phi \mapsto {m'}^*\Om\, , \quad m'\in \M'\, .
\een
Following \cite{Ar76,Ar77}, we take the closures of these operators and consider their polar decomposition, 
yielding the relative modular operators $(\Delta)$ and conjugations $(J)$:
\[
S_{\Om,\Phi;\M} = J_{\Om,\Phi}\Delta^{1/2}_{\Om,\Phi}\, ,\qquad S_{\Om,\Phi;\M'} = J'_{\Om,\Phi}\Delta'^{1/2}_{\Om,\Phi}\, ,
\]
where $\Delta_{\Om,\Phi}, \Delta'_{\Om,\Phi}$ are self-adjoint and positive, and 
$J_{\Om,\Phi}, J'_{\Om,\Phi}$ are anti-unitary.
When the two vectors are identical, $\Phi=\Om$, we  write $\Delta_{\Om,\Om} = \Delta_\Om$, etc.
The two vectors $\Omega, \Phi$ define normal state functionals  $\f = (\Phi, \cdot \Phi)$, $\om = (\Om, \cdot \Om)$ on $\M$. 
Their associated Connes-cocycle \cite{C73}, $u_s$, is related to the relative modular operators by \cite{Ar76,Ar77}
\ben\label{uDD}
u_s 
= \Delta^{is}_{\Om,\Phi}\Delta^{-is}_{\Phi} =
\Delta^{is}_{\Om}\Delta^{-is}_{\Phi,\Om}
\, ,
\een
and similarly for $\M'$ and its modular operators. 

It is known \cite{C73} that $u_s$ respectively $u_s'$ are unitary operators from $\M$ respectively $\M'$ for all $s \in \RR$.
These constructions may be generalized when $\Phi,\Omega$ are not cyclic and separating, 
see \cite{Ar77}, \cite[App. C]{AM82}. For simplicity, we will avoid such considerations in this paper by generally assuming that $\Phi,\Omega$ are 
cyclic and separating.

In the context of general von Neumann algebras as appropriate for us, the relative entropy is defined by \cite{Ar76,Ar77} 
\ben
\label{Sdef}
S(\Phi|\!| \Om)_\M =  - (\Phi , \log \Delta_{\Om,\Phi}\Phi). 
\een
and depends on $\M$ through the definition of the relative modular operator. The work by \cite{Ar76,Ar77} shows that
$S(\Phi|\!| \Om)_\M \equiv S(\varphi|\!| \omega)_\M$, where we intend to mean that $S$ can be regarded as a
functional of  $\f = (\Phi, \cdot \Phi)$, $\om = (\Om, \cdot \Om)$, i.e., the normal state functionals on $\M$ associated with the vectors $\Phi,\Om \in \H$.
A similar statement holds when replacing $\M$ by its commutant. See e.g., \cite{OP} for further discussion of $S$.

\subsection{Half-sided modular inclusions and ant formula}
\label{hsm}

Let $\N \subset \M$ be a proper inclusion of von Neumann algebras acting on $\H$ and suppose the following:
\begin{enumerate}
\item We have a cyclic and separating vector $\Omega \in \H$ for both $\M$ and $\N$ . 
\item We have that $\Delta^{-it}_\Omega \N \Delta^{it}_\Omega \subset \N$ for any $t \ge 0$, with $\Delta_\Omega$ the modular operator 
associated with $\M$. 
\end{enumerate}
The structure $(\N \subset \M, \Omega)$ subject to these conditions is called a half-sided modular inclusion (HSMI) \cite{Wies1}. 

Wiesbrock's theorem \cite{Wies1,AZ05,Florig} for HSMIs $(\N\subset\M, \Omega)$ states that
there is a strongly continuous 1-parameter group of unitary operators $U(a), a \in \RR$ on $\H$ such that
 \begin{enumerate}
 \item $U(a)=e^{iaP}$ with self-adjoint $P$ such that  ${\rm spec} P  \subset [0,\infty)$.
 \item $U(a) \Omega = \Omega$ for all $a \in \RR$.
 \item $\M(a) := U(a) \M U(a)^* \subset \M$ and $\M(-a)'=U(-a) \M' U(-a)^* \subset \M'$ for all $a \ge 0$.
 \item $\M(1) = \N$.
 \item If $b > a$ , then $(\M(b) \subset \M(a),\Omega)$ are HSMIs, and 
 if $a>b$  $(\M(-a)' \subset \M(-b)',\Omega)$ are  HSMIs.
 \item $\Delta^{-it}_{\Omega} U(a) \Delta^{it}_{\Omega} = U(e^{2\pi t} a)$ for all $t,a \in \RR$, 
 or equivalently $\Delta^{-it}_{\Omega} P \Delta^{it}_{\Omega} = e^{2\pi t} P$ for all $t \in \RR$, on the domain $\D(P)$ of $P$, i.e., 
 $\D(P)$ is in particular invariant under $\Delta^{it}_{\Omega}$. As a consequence, $\Delta_\Omega^{-it} \M(a) \Delta_\Omega^{it} =\M(e^{2\pi t}a)$.
 \item $J_{\Omega} U(a) J_{\Omega} = U(-a)$ for all $a \in \RR$.
 \end{enumerate}
 Given $\Phi \in \H$, one may consider the relative entropies
with respect to $\M(a)$ or $\M(a)'$ 
\ben
\label{Sadef}
S(a):=S(\Phi|\!| \Om)_{\M(a)}, \quad
\bar S(a):=S(\Phi|\!| \Om)_{\M(a)'}
\een
as functions of $a$.

Suppose that the derivative $\partial S(a)$ of this function exists for some $a \in \R$, 
$\Phi \in \D(P) \cap \D(\log \Delta'_a)$, 
and suppose that $u_{s}' \Phi \in \D(P)$ for sufficiently large $s$. Here, 
$u_s' $ is the Connes-cocycle \eqref{uDD} and 
 $\Delta'_a$ is the relative modular operator (both for  $\M(a)'$ and $\Phi,\Omega$). Then we have the ``ant formula''
\cite{HL25,CF18}
\ben
\label{antleft}
-\partial S(a) = 2\pi \inf_{u' \in \M(a)'} (u'\Phi, P\,  u'\Phi) = 2\pi \lim_{s \to +\infty} (u'_s\Phi, P \, u'_s\Phi),
\een
\textcolor{black}{where $u' \in \M(a)'$ is unitary.} Analogously, suppose that $\partial \bar S(a)$ 
exists, 
$\Phi \in \D(P) \cap \D(\log \Delta_a)$, 
and $u_{s} \Phi \in \D(P)$ for a sufficiently small $s$. Now, $u_s$ is the Connes-cocycle \eqref{uDD} for $\M(a)$ and $\Phi,\Omega$. 
Since this situation is related to the previous one by the modular conjugation $J$, one has \cite{HL25,CF18}
\ben
\label{antright}
+\partial \bar S(a) = 2\pi \inf_{u \in \M(a)\,} (u \, \Phi, P \, u\, \Phi) = 2\pi \lim_{s \to -\infty} (u_s\Phi, P \, u_s\Phi) ,
\een
\textcolor{black}{where $u \in \M(a)$ is unitary.} 

The interpretation of the limit on the right side of \eqref{antright} is as follows in a context where there are non-trivial 
$m \in \M(a) \cap \M(b)', a<b$.
\ben
\label{ergodic}
\lim_{s \to -\infty} (u_s\Phi, m u_s\Phi) 
\lim_{s \to -\infty} (\Delta^{is}_{\Omega} \Phi, m \Delta^{is}_{\Omega} \Phi) = (\Omega, m \Omega) \|\Phi \|^2, 
\een
using the ergodic nature of the modular flow $\Delta_\Omega^{is}$ of $(\M(a), \Omega)$ \cite{Roberts} in the last step and the definition of $u_s$ and properties of the modular 
relative operators in the first. 

Thus, with respect to $\M(a) \cap \M(b)'$, the state 
functional $\varphi_s = (u_s\Phi, \, . \, u_s\Phi)$ tends to a multiple of the vacuum functional $\omega = (\Omega, \, . \, \Omega)$. 
In this sense, 
$u_s$ ``scales away'' the departure of $\Phi$ from the vacuum $\Omega$. Similar remarks apply to 
the expression on the right side of the ant formula \eqref{antleft} for the commutant. 

\section{QDEC in $d=2$}

\subsection{Formulation and proof}

For simplicity, we first consider such theories in $d=2$ dimensions. Null coordinates on $\RR^2$ are defined as $x^\pm = (x^0 \pm x^1)/\sqrt{2}$. The Minkowski metric 
$\eta_{\mu\nu}, \mu, \nu = 0,1,$ is given by $\eta_{11} = -\eta_{00} = 1, \eta_{01} = \eta_{10} = 0$, or equivalently by $\eta_{+-} = \eta_{-+} = -1, \eta_{++}=\eta_{--}=0$. 
A (right) wedge with apex $x=(x^0,x^1) \in \RR^2$ is the domain
\ben
\label{Wdef0}
W(x) = \{ y \in \RR^2 \, : \, y^+>x^+, y^-<x^-\}. 
\een
The null components of the energy momentum operator are denoted by $(P_+, P_-)$, and the vacuum state is denoted by $\Omega$. Then $P_\pm$ are self-adjoint, positive operators 
and $\Omega$ is the unique eigenvector with eigenvalue $0$. They give a strongly continuous representation $U(a) = e^{ia^+ P_+ + ia^- P_-}$ acting geometrically 
on the Haag-Kastler net \cite{haag_2} $\{ \A(O) \}$ defining the QFT by
\ben
U(a) \A(O) U(a)^* = \A(O + a). 
\een
Then, assuming a sufficiently strong version of the Araki-Kastler axioms \cite{haag_2} as we will from now, $(\N:=\A(W(x^+=1,x^-=0)) \subset \M:= \A(W(0)), \Omega)$ is a 
HSMI such that $P = P_+$ \cite{Bor1,Bor2}. Likewise $(\N:=\A(W(x^+=0,x^-=1))' \subset \M:=\A(W(0))', \Omega)$ is a 
HSMI such that $P=P_-$. 

To the first of these HSMIs, we now apply \eqref{antleft}, whereas to the second we apply \eqref{antright}, and then we add these two formulas 
with non-negative weights $\tilde r^+, \tilde r^- \ge 0$, and we put $r^+ = \tilde r^+, r^-=-\tilde r^-$, so that 
\ben
\label{dual}
r^\mu = \epsilon^\mu{}_\nu \tilde r^\nu,
\een
$\epsilon_{\mu\nu}$ is the up to normalization unique antisymmetric tensor on $\RR^2$. Our normalization is $\epsilon_{-+}=-1$, 
so that $\tilde r$ is a future pointing timelike or null vector and $r$ is a right pointing spacelike or null vector. Then we obtain
\ben
\label{antboth}
\begin{split}
&-(r^+ \partial_+ + r^- \partial_-) S(x)\\ 
=& 2\pi \inf \Big\{ \tilde r^+ (u_+'\Phi, P_+ u_+'\Phi) + \tilde r^- (u_-'\Phi, P_- u_-'\Phi) :
u_\pm' \in \A(W(x))', \, u'_\pm \, {\rm unitary} \Big\} ,
\end{split}
\een
where 
\ben
\label{Sx}
S(x) := S(\Phi |\! | \Omega)_{\A(W(x))}. 
\een
The formula \eqref{antboth} holds under the assumptions that the derivatives exist, that 
$\Phi \in \D(P_+) \cap \D(P_-) \cap \D(\log \Delta_{\Phi,\Omega;\A(W(x))'})$, $u'_t \Phi \in \D(P_+) \cap \D(P_-)$ for all $t \in \RR$, where $u_t'$ is the Connes-cocycle \eqref{uDD} for $\A(W(x))'$. Let us now suppose these assumptions 
hold for all $x$ in a fixed wedge, say $x \in W(0)$. 

If $s \in W(0)$, i.e., $s$ is a spacelike vector pointing to the right, then obviously $\A(W(x+s)) \subset \A(W(x))$. As a consequence, 
the variational nature of \eqref{antboth} gives 
\ben
\label{antboth1}
-r^\mu \partial_\mu S(x) \ge  -r^\mu \partial_\mu S(x+s),
\een
for all states $\Phi$ satisfying the above domain assumptions. Integrating this inequality along the spacelike line with tangent $r$ and starting point $x$ (lying within $W(x)$), 
we thereby get the following proposition.

\begin{proposition}
\label{propant}
Suppose that that $S(x)$ \eqref{Sx}  is once differentiable for all $x \in W(0)$, and that for all such $x$, 
$\Phi \in \D(P_+) \cap \D(P_-) \cap \D(\log \Delta_{\Phi,\Omega;\A(W(x))'})$, $u'_t \Phi \in \D(P_+) \cap \D(P_-)$ for all $t \in \RR$, where $u_t'$ is the Connes-cocycle \eqref{uDD} for $\A(W(x))'$. Then
\ben
\label{Sconv}
S(x+r+s)+S(x) \ge S(x+r)+S(x+s) \qquad \text{for \ all $x,r,s \in W(0)$}.
\een
\end{proposition}

\begin{remark}
Note that if $S(x)$ is twice differentiable in $x\in W(0)$, the condition \eqref{Sconv} may be restated equivalently as 
\ben
\label{QDECdiff}
r^\mu s^\nu \partial_\mu \partial_\nu S(x) \ge 0 \qquad \text{forall $x,r,s \in W(0)$.}
\een
\end{remark}

\textcolor{black}{Following \cite[Sec. 6]{HL25b},} we now generalize this proposition to vectors $\Phi$ in a more natural domain. For this, we recall that $M' \hat\in \A(W(0))'$ means that $M'=A'u'$
where $A'$ is a self-adjoint operator such that its spectral projections for any bounded interval belong to $\A(W(0))'$, and such that $u' \in \A(W(0))'$ is unitary. (In particular, 
all $m' \in \A(W(0))'$ \textcolor{black}{satisfy such a condition}.)

\begin{theorem}
\label{QDEC}
Let $\Phi = M'\Omega \in \H$ in $S(x)$ \eqref{Sx}, where $M' \hat\in \A(W(0))'$. Then \eqref{Sconv} holds for any $x,r,s \in W(0)$.
\end{theorem}

\begin{proof}
First let  $\varphi$ be a state on $\A(W(0))$ such that, for some $c>0$, we have that $c^{-1}\varphi \le \omega \le c\varphi$, where $\omega = (\Om, \ . \ \Om)$ is the 
state on $\A(W(0))$ induced by the vacuum vector. 
Let $T_\lambda^\pm, \lambda>0$ be the normal, unital, completely positive linear maps (i.e., channels)  $\A(W(0)) \to \A(W(0))$ defined by 
\ben
\label{Tladef}
T_\lambda^\pm(m) = \lambda \int_{(0,\infty)} e^{-\lambda a} e^{\pm iaP_\pm} m e^{\mp iaP_\pm} \, da. 
\een
Then $[T_\lambda^+, T_\lambda^-]=0$, and $T_\lambda^\pm$ restrict to channels on $\A(W(x))$ for each $x \in W(0)$. 
Define $\varphi_\lambda = \varphi \circ T_\lambda^+T^-_\lambda$. Then $c^{-1}\varphi_\lambda \le \omega \le c\varphi_\lambda$
as states on $\A(W(x))$. Let $\Phi_{x,\lambda}$ be the unique vector in the natural 
cone of $\A(W(x))$ defined by $\Omega$ representing this state. Then by a minor generalization of 
\cite[Thm. A.1]{HL25b}, we have $\Phi_{\lambda,x} \in \D(P_+) \cap \D(P_-) \cap \D(\log \Delta_{\Phi,\Omega;\A(W(x))'})$
and $u'_{\lambda, x, t} \Phi_{\lambda,x} \in \D(P_+) \cap \D(P_-)$ for all $t \in \RR$, where $u_{\lambda, x, t}'$ is the Connes-cocycle \eqref{uDD} for $\A(W(x))$, and the pair 
of states $\Omega, \Phi_{\lambda, x}$. As a consequence of \cite[Lem. 4.5]{HL25b} it follows that 
\ben
\label{rSla}
S_\lambda(x) := S(\varphi_\lambda |\! | \omega)_{\A(W(x))}
\een
is continuously differentiable for $x \in W(0)$. By \cite[Thm. 5.1, Rem. 5.2]{HL25b} applied to $\M := \A(W(x))$, we therefore have \eqref{antboth} for 
{\it any} representative $\Phi_{x,\lambda}$ of $\varphi_\lambda$ on $\A(W(x))$, and as a consequence we get \eqref{Sconv} for $S_\lambda$. 

Since by \cite[Thm. A.1]{HL25b}, $\lim_{\lambda \to 0} T_\lambda(m) \Phi = \Phi$ in the strong sense for all $m \in \A(W(0))$, and since $T_\lambda^+T^-_\lambda$ is a channel which leaves 
$\omega$ invariant,  
the lower semi-continuity \cite{Ar76,Ar77}  and monotonicity \cite{U77} of the relative entropy \eqref{rSla} give 
\ben
S(x) \le \liminf_{\lambda \to 0+} S_\lambda(x) \le \limsup_{\lambda \to 0+} S_\lambda(x) \le S(x), 
\een
so $\lim_{\lambda \to 0+} S_\lambda(x) = S(x)$. Thus, we get \eqref{Sconv} for $S$ and all $x,r,s \in W(0)$ for all 
states satisfying $c^{-1}\varphi \le \omega \le c\varphi$ as functionals on $\A(W(0))$.

It is well-known that such states are of the form $\varphi = (m'\Omega, \ . \ m'\Omega)$ for some $m' \in \A(W(0))'$. By applying the 
same arguments as in proofs of \cite[Thm. 6.3]{HL25b} to this class of states we get the statement of the theorem.
\end{proof}

\subsection{QDEC for coherent states of a free scalar field}
\label{coherent}

We now illustrate the QDEC \eqref{Sconv} for coherent states of a free hermitian scalar QFT of mass $m > 0$ on $d=2$ dimensional 
Minkowski spacetime. The algebraic construction of this theory in $d$ dimensions can be described as follows, see, e.g., \cite{sanders}
for details. $V(F), F \in C^\infty_0(O,\RR)$, $O\subset \RR^d$ open with compact closure, are the 
generators of the local Weyl $C^*$-algebras ${\mathfrak A}(O)$ for this QFT, i.e., ${\mathfrak A}(O)$ is the norm closure 
of the unique $*$-algebra generated by the unitary elements $V(F)$. They are subject to the Weyl form of the 
canonical commutation relations:
\ben
V(F) V(G) = {\rm exp}\left[ -i \, {\rm Im} \, w(F,G) \right] V(F+G), \quad V(F)^* = V(-F). 
\een
Here $w$ corresponds to the Wightman 2-point function of the theory, 
\ben
w(F,G) = 2\pi \int\limits_{\RR^{d-1}} \frac{dp}{2E_p} \overline{\hat F(E_p,p)} \hat G(E_p,p), 
\een
where a hat denotes the Fourier transform\footnote{Our convention for the Fourier transformation in $\RR^n$ is 
$\hat F(k) = (2\pi)^{-n/2} \int e^{ikx} F(x) d^n x$.}, and $E_p = \sqrt{|p|^2+m^2}$. Then $\omega(V(F)) := e^{-w(F,F)/2}$ defines a 
state on the closure of $\cup_{O \subset \RR^d} {\mathfrak A}(O)$ in the unique $C^*$-norm. In the GNS-representation $\pi_\omega$ of this state, we define 
the local von Neumann algebras as $\A(O):= \pi_\omega[{\mathfrak A}(O)]''$. 

The GNS-vector $\Omega \equiv \Omega_\omega$ of $\omega$ is the vacuum state, and $\{\A(O)\}$ carries 
a positive energy representation of the Poincar\'e group $SO_0(1,d-1) \ltimes \RR^d$ leaving $\Omega$ invariant. 
 $\{\A(O)\}$ is a net satisfying a sufficiently strong version of the Haag-Kastler axioms such that 
Borchers' theorem \cite{Bor1,Bor2} holds. Therefore, we have the QDEC
in the sense of Thm. \ref{QDEC} in $d=2$. 

We now illustrate this statement for the case of coherent states, defined as follows. Let $E(F,G):=2{\rm Im} \,w(F,G)$, which is a bi-distribution 
on $\RR^2 \times \RR^2$ in the Schwarz class. It is the Schwarz kernel of a linear operator $E: C^\infty_0(\RR^2) \to C^\infty(\RR^2)$, which is equal to the advanced minus 
retarded fundamental solution of the Klein-Gordon (KG) equation.
The image $EF=:\phi$ is a smooth solution to the KG equation, 
\ben
\label{KG}
(-\square +m^2) \phi(x) = 0,
\een
where $\square=\eta^{\mu\nu}\partial_\mu\partial_\nu$.
Then $\Phi := V(F) \Omega$ defines a state called a coherent state associated with the classical solution $\phi$. 
The following expression is a covariantized form of the expression for the relative entropy $S(x)$ \eqref{Sx} rigorously proven in \cite{Ciolli:2019mjo}:
\ben
\label{Scoh}
S(x) = 2\pi \int_{\Sigma(x)} (y-x)^\mu \epsilon_{\mu}{}^\nu T_{\nu\sigma} d\Sigma^\sigma(y).
\een
Here, $\Sigma(x)$ is any Cauchy surface of the right wedge $W(x)$ with apex $x$, $d\Sigma^\mu$ is the induced integration element on that Cauchy
surface, and $\epsilon_{\mu\nu}$ is the up to normalization unique antisymmetric tensor on $\RR^2$. 
$T_{\mu\nu}$ is the stress energy tensor for the classical 
solution $\phi$, 
\ben
\label{T}
T_{\mu\nu} = \partial_\mu \phi \partial_\nu \phi - \frac{1}{2} \eta_{\mu\nu} \left( \partial^\sigma \phi \partial_\sigma \phi + m^2 \phi^2 \right). 
\een
As a consequence of \eqref{KG}, we have $\partial^\mu T_{\mu\nu} = 0$.
Taking a derivative $r^\mu \partial_\mu$ along the spacelike vector $r \in W(0)$, we get
\ben
r^\mu \partial_\mu S(x) = -2\pi \int_{\Sigma(x)} r^\mu \epsilon_{\mu}{}^\nu T_{\nu\sigma} d\Sigma^\sigma(y)
\een
from \eqref{Scoh}. In order to take a further derivative $s^\mu \partial_\mu$ along a second spacelike vector $s \in W(0)$, it is 
convenient to assume that $s^\mu s_\mu = 1$, and that $\Sigma(x)$ is given by $\Sigma(x) = \{ x+\rho s : \rho>0\}$. Then 
the future directed unit normal to $\Sigma(x)$ is $u^\nu = -\epsilon^\nu{}_{\mu} s^\mu$, and $d\Sigma^\mu = u^\mu d\rho$. 
Since the choice of the Cauchy surface $\Sigma(x)$ of $W(x)$ is immaterial due to $\partial^\mu T_{\mu\nu} = 0$ and the fact that ${\rm supp} \phi$ has a compact 
intersection with any Cauchy surface, we can therefore write 
\ben
r^\mu \partial_\mu S(x) = -2\pi \int_{0}^\infty u^\sigma r^\mu \epsilon_{\mu}{}^\nu T_{\nu\sigma}(x+\rho s) d\rho.
\een
This gives
\ben
r^\mu s^\nu \partial_\nu \partial_\mu S(x) = 2\pi \, r^\mu s^\nu \epsilon_{\mu}{}^\alpha \epsilon_\nu{}^\beta T_{\alpha\beta}(x)
\een
for all $x \in \RR^2$.
The dual vectors $\tilde r$ and $\tilde s$ \eqref{dual} are timelike and future directed.
Thus $T_{\mu\nu} \tilde r^\mu \tilde s^\mu \ge 0$, because the DEC holds for the classical stress energy tensor \eqref{T}.
We thereby confirm the QDEC in differential form \eqref{QDECdiff} for coherent states--in fact the second derivative of the relative entropy 
is just the expression for the classical DEC for $\phi$:
\ben
r^\mu s^\nu \partial_\nu \partial_\mu S(x) = 2\pi \, \tilde r^\mu \tilde s^\nu T_{\alpha\beta}(x) \ge 0.
\een
This result  is consistent with 
Thm. \ref{QDEC} even tough the assumptions of this theorem do not quite apply to the coherent state vector $\Phi = V(F)\Omega$.
Based on this example we therefore conjecture that Thm. \ref{QDEC} holds for {\it all} states $\Phi$ whenever the terms 
on the right side of \eqref{Sconv} are finite. 

\subsection{QDEC and state recovery}
\label{staterecovery}

Let $T: \N \to \M$ be a 2-positive, normal, unital map between von Neumann algebras $\M,\N$, i.e., a (2-positive) channel. Let $\omega$ be a 
faithful state on $\M$ and assume that also $\omega \circ T$ is a faithful state on $\N$. Let $\Omega_\N$ and $\Omega_\M$ be the 
state representers of $\omega \circ T$ and $\omega$, respectively, in the corresponding GNS Hilbert spaces. We denote the 
corresponding modular operators by $\Delta_{\Omega, \N}$ and $\Delta_{\Omega, \M}$, respectively. 

We can use the state $\Omega_\M$ to define the ``KMS scalar product'' on $\M$, by 
\ben
\langle m_1,  m_2 \rangle_\Omega := (\Omega_\N, m_1^* \Delta_{\Omega, \M}^{1/2} m_2 \Omega_\M), 
\een
and we can similarly define a KMS inner product on $\N$ using $\Omega_\N$ and $\Delta_{\Omega,\N}$. Then the adjoint 
of $T$ with respect to these KMS scalar products is a 2-positive channel $T^\dagger: \M \to \N$ called the Petz map (of $T$).
The rotated, smeared Petz map is defined by 
\ben
\rho(m) := \int_{-\infty}^\infty \beta(t) \Delta_{\Om,\N}^{-it} T^\dagger \left( \Delta_{\Om,\M}^{it} m \Delta_{\Om,\M}^{-it} \right) \Delta_{\Om,\N}^{it} \, dt
\een
where $\beta(t) := \pi [1+\cosh(2\pi t)]^{-1}$ is a probability density. $\rho$ is a 2-positive channel $\rho: \M \to \N$. In \cite{hollandspetz}, the following 
``recovery bound'' was proven for a arbitrary normal state $\varphi$ on $\N$:
\ben
\label{recovery:1}
S(\varphi |\! | \omega)_{\M}-S(\varphi \circ T |\! | \omega \circ T)_{\N} \ge -\log F(\varphi | \! | \varphi \circ T \circ \rho)_\M. 
\een
Here, $F(\varphi | \! | \psi)_\N$ is the fidelity between two normal states $\psi, \varphi$ on the von Neumann algebra $\M$. It can be seen as a transition 
amplitude between the states and is defined as 
\ben
F(\varphi | \! | \psi)_\M = \sup \{ |(\Psi, \Phi)|^2 \ : \ \text{vector representatives $\Phi,\Psi$ of $\varphi,\psi$} \}.
\een
The fidelity satisfies $0 \le F(\varphi | \! | \psi)_\M \le 1$, with $F(\varphi | \! | \psi)_\M=1$ if and only if $\psi = \varphi$ as states on $\M$. Thus, the right side of the recovery bound 
\eqref{recovery:1} is a non-negative number expressing an improvement of the data processing inequality for the relative entropy. 

We will now combine the state recovery bound with the QDEC. For this, we set $\N:=\A(W(R/\sqrt{2}, R/\sqrt{2}))$ and $\M=\A(W(0))$ for some $R\ge 0$, i.e., $\N \subset \M$ is a von Neumann subalgebra, and we define $S(R/\sqrt{2},R/\sqrt{2})$ as before in \eqref{Sx}, see fig. \ref{fig:3}.

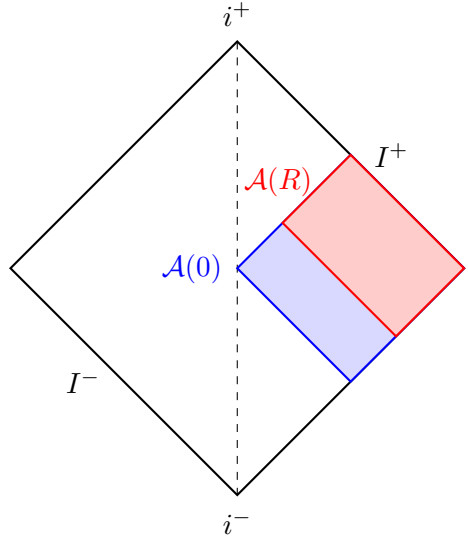
\begin{figure}[t]
\centering
\begin{tikzpicture}[scale=3]

\draw[double distance=2pt]
(0,1)--(1,0)--(0,-1)--(-1,0)--cycle;
\draw[dashed] (0,1)--(0,-1);
\node at (0,1.12) {$i^+$};
\node at (0,-1.12) {$i^-$};
\node at (-.7,-.52) {${\mathscr I}^-$};
\node at (.73,0.55) {${\mathscr I}^+$};
\node at (0.12,0) {$i^0$};
\draw[-{Stealth[length=2mm]}] (.04,.04)--(-.26,.34);
\draw[-{Stealth[length=2mm]}] (.06,.06)--(-.24,.36);
\draw[-{Stealth[length=2mm]}] (.08,.08)--(-.22,.38);
\draw[-{Stealth[length=2mm]}] (.1,.1)--(-.2,.4);
\draw[-{Stealth[length=2mm]}] (.12,.12)--(-.18,.42);
\draw[-{Stealth[length=2mm]}] (.14,.14)--(-.16,.44);
\draw[-{Stealth[length=2mm]}] (.16,.16)--(-.14,.46);
\fill[blue!15]
(0,0)--(0.5,.5)--(1,0)--(0.5,-.5)--cycle;

\draw[blue,thick]
(0,0)--(0.5,.5)--(1,0)--(0.5,-.5)--cycle;

\node[blue] at (-.2,0) {$\A(0)$};

\fill[red!20]
(0.2,0.2)--(0.5,.5)--(1,0)--(0.7,-.3)--cycle;

\draw[red,thick]
(0.2,0.2)--(0.5,.5)--(1,0)--(0.7,-.3)--cycle;

\fill (0,1) circle (1pt);
\fill (0,-1) circle (1pt);
\fill (1,0) circle (1pt);
\fill (-1,0) circle (1pt);

\node[red] at (0.65,0.0) {$\A(R)$};

\end{tikzpicture}
\caption{Algebras of $W(0)$ (blue) and of $W(R/\sqrt{2},R/\sqrt{2})$ (red) in a Penrose diagram of 2-dimensional Minkowski spacetime.
The arrows indicate radiation leaving the blue wedge through the upper horizon. 
}
\label{fig:3}
\end{figure}
We consider the embedding $\iota_R: \N \to \M$ as a (completely positive) channel, and we let $\rho_R$ be the rotated Petz channel of $\iota_R$. The recovery bound reads
\ben
\label{recovery:2}
S(0,0)-S(R/\sqrt{2},R/\sqrt{2}) \ge -\log F(\varphi | \! | \varphi_{|\A(R)} \circ \rho_R). 
\een
in our previous notations. Here $\varphi_{|\A(R)}$, is the ``partial state'' functional induced by the vector $\Phi \in \H$ on the subalgebra 
$\A(R):=\A(W(R/\sqrt{2},R/\sqrt{2}))$ of $\A(0):=\A(W(0))$, and $\varphi$ is the state function on $\A(0)$.
A more explicit form of the right side via a formula for $\rho_R$ using the null-translation group
given by the unitary positive energy representation of the QFT 
follows from \cite[Cor. 2]{hollandspetz}.

The following theorem expresses that the ability to recover $\varphi$ from $\varphi_{|\A(R)}$ is guaranteed to be greater if the state $\Phi$ has less energy.

\begin{theorem}
\label{thm:recov}
Let $R \ge 0$, $\Phi = M'\Omega \in \H$ in \eqref{Sx}, where $M' \hat\in \A(W(0))'$, and let $\varphi_{|\A(R)}$ be the partial state functional on
$\A(R):=\A(W(R/\sqrt{2},R/\sqrt{2})) \subset \A(W(0))$ induced by $\Phi$. Then 
\ben
\label{recovb}
-\log F(\varphi | \! | \varphi_{|\A(R)} \circ \rho_R) \le 2\pi R \ (\Phi, P_0\Phi),
\een
where $P_0$ is the Hamiltonian (generator of time translations), and $\rho_R$ is the rotated Petz recovery channel.
\end{theorem}

\begin{remark}
Since the fidelity only depends on the state functional $\varphi = (\Phi, \, . \, \Phi)$ on $\A(W(0))$ induced by $\Phi$, we could optimize the bound replacing the right side by 
the infimum over all state representatives, i.e., by $2\pi R  \inf \{ (u'\Phi, P_0u'\Phi) :  u' \in \A(W(0))'  \text{unitary}\}$.
\end{remark}
\begin{remark}
Using the form of the recovery bound given in  \cite[Cor. 2]{hollandspetz}, we may express our bound also as 
\ben
2\pi \ (\Phi, P_0\Phi) \ge -\int\limits_R^\infty \frac{da}{a^2} \, \log F( \varphi | \! | e^{iaP_+} \varphi e^{-iaP_+} ).
\een
\end{remark}

\begin{proof}
(of Thm. \ref{thm:recov})
At first let $\Phi = m'\Omega$, where $m' \in \A(W(0))'$. Then the corresponding state functional $\varphi = (\Phi, \ . \ \Phi)$ is 
$c$-comparable to the state functional $\omega = (\Omega, \ . \ \Omega)$ on any algebra $\A(W(x))$ for $x \in W(0)$, i.e.
we have that $c^{-1}\varphi \le \omega \le c\varphi$. 
Let $T_\lambda^\pm, \lambda>0$ be the normal, unital, completely positive linear maps (i.e., channels)  $\A(W(0)) \to \A(W(0))$ defined by 
\eqref{Tladef}. Again, we define $\varphi_\lambda = \varphi \circ T_\lambda^+T^-_\lambda$. 

Then as already argued in the proof of Thm. 
\ref{QDEC}, it follows that $S_\lambda(x)$ given by 
\eqref{rSla} is differentiable and we have \eqref{antboth} for $x= \in W(0)$. Let $\tilde r$ be a future pointing timelike or null 
vector, i.e. $\tilde r^\pm \ge 0$, and $r^\mu = \epsilon^\mu{}_\nu \tilde r^\nu$, so that $r$ is a right pointing spacelike or null 
vector, i.e., $\pm r^\pm \ge 0$  (recall the double null coordinates $x^\pm=(x^0 \pm x^1)/\sqrt{2}$). Integrating \eqref{antboth} up, we have
\ben
\begin{split}
& S_\lambda(x-r) - S_\lambda(x) = \ \int_0^1 \frac{d}{dt} S_\lambda(x-tr) \ dt \\
=& \ 2\pi \int_0^1 \inf \bigg\{ 
\tilde r^+(u_+'\Phi, P_+ u_+' \Phi) + \tilde r^-(u_-'\Phi, P_- u_-' \Phi) \ : 
  u'_\pm \in \A(W(x-tr))'
\bigg\} \ dt \\
\le & \ 2\pi \, \bigg[ \tilde r^+(\Phi, P_+ \Phi) + \tilde r^-(\Phi, P_-\Phi) \bigg]
\end{split}
\een
Taking $r=x=(0,R)$ then gives that
\ben
\label{Sdiffbound}
S_\lambda(0,0)-S_\lambda(0,R) \le 2\pi R \, (\Phi, P_0\Phi),
\een
where the arguments of $S_\lambda$  refer to the ordinary Cartesian coordinates $x=(x^0,x^1)$ and not null coordinates.
By the same argument as in the proof of Thm. 6.1 of \cite{HL25b}, $S_\lambda(x) \to S(x)$ for $x \in W(0)$.  
So \eqref{Sdiffbound} holds for $S$ instead of $S_\lambda$ as well.
Combining this statement with the recovery bound \eqref{recovery:2} and 
monotonicity of the relative entropy \cite{U77}, $S(R/\sqrt{2},R/\sqrt{2}) \ge S(0,R)$, concludes the proof when $\Phi$ 
is of the form $m'\Omega, m' \in \A(W(0))'$.

To generalize this to the larger domain of states given in the statement, we use the same approximation steps as in the proof of Thm. 6.3 of \cite{HL25b}:
Let $e_n'$ be the spectral projections for the spectral interval $[0,n]$ of $|M'|$, which are elements of $\A(W(0))'$ since $M'$ is affiliated. Then we set
$\Phi_{\epsilon,n} := (\epsilon 1 + |M'|^2e'_n)^{1/2} \Omega$ for $\epsilon>0$. Then each state $\Phi_{\epsilon,n}$ is of the form $m'_{\epsilon,n} \Omega$ 
for $m_{\epsilon,n}' \in \A(W(0))'$.
We define $S_{n,\epsilon}(x)$ by \eqref{rSla} but with the state $\Phi_{\epsilon,n}$
instead of $\Phi$. By our previous arguments, we have \eqref{Sdiffbound}, i.e., 
\ben
\label{Sdiffbound1}
S_{\epsilon,n}(0,0)-S_{\epsilon,n}(0,R) \le 2\pi R \, (\Phi_{\epsilon,n}, P_0 \, \Phi_{\epsilon,n}).
\een
We now take the limit as $n \to \infty, \epsilon \to 0$. On the left side, we can use the fact (proof of Thm. 6.3 of \cite{HL25b}) $\lim_{\epsilon \to 0} \lim_{n \to \infty} S_{n,\epsilon}(x) = S(x)$ for $x \in W(0)$.
 On the right side we can use that $\lim_{\epsilon \to 0} \lim_{n \to \infty} \Phi_{\epsilon,n} = \Phi$ (convergence in norm). 
 Then, if $\Phi \in \D(P_0^{1/2})$, we get 
\ben
\label{Sdiffbound2}
S(0,0)-S(0,R) \le 2\pi R \, (\Phi, P_0\, \Phi), 
\een
whereas if $\Phi \notin \D(P_0^{1/2})$, then the right side is $+\infty$ and there is nothing to show. 
Finally, we apply the recovery bound \eqref{recovery:2}, which has no domain restrictions, to \eqref{Sdiffbound2}.
\end{proof}

\section{QDEC in $d>2$ dimensions}

\subsection{Rigorous formulation}
We begin by splitting the coordinates of $d$-dimensional Minkowski space into $x = (x_\perp, x_\parallel),$ where $x_\perp = (x^0, x^1)$ and 
$x_\parallel = (x^2, \dots, x^{d-1})$. We fix a function $f \in C^\infty_0(\RR^{d-2}_\parallel)$ such that $f(x_\parallel) \ge 0$, and a 2-vector $r_\perp = (r_\perp^+,r_\perp^-)$ with $r_\perp^\pm>0$, and define the 
deformed wedges
\ben
\label{Wdef}
W_f(x_\perp, r_\perp) = \{ y \in \RR^{d} \, : \,  y^+ > x^+ + r_\perp^+ f(y_\parallel), y^- < x^- +r_\perp^- f(y_\parallel) \}, 
\een
see fig. \ref{fig1}.
As before null coordinates on $\RR^2_\perp$ 
are defined as $x^\pm = (x^0 \pm x^1)/\sqrt{2}$, so that the Minkowski metric is 
$\eta = -2dx^+ dx^- + dx_\parallel^2$.
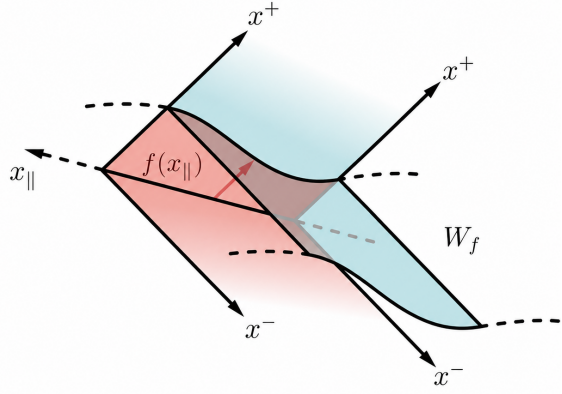
\begin{figure}
\centering
\begin{tikzpicture}[
    line width=0.9pt,
    line cap=round,
    line join=round,
    scale=.8
]


\coordinate (P) at (0,0);
\coordinate (P+) at (1,0);
\coordinate (P-) at (-0.5,0);
\coordinate (P--) at (-1.5,0);
\coordinate (Q) at (8,0);

\coordinate (V) at (13,1);

\coordinate (A) at (3.0,3.00);
\coordinate (A+) at (3.3,3.30);
\coordinate (A++) at (4,4.0);
\coordinate (B) at (11.0,3.00);

\coordinate (C) at (3.00,-3);
\coordinate (C-) at (3.30,-3.3);
\coordinate (C--) at (4,-4);
\coordinate (D) at (11.0,-3);

\coordinate (U1) at (2.00,2);
\coordinate (U1+) at (3.00,2);
\coordinate (U1++) at (3.10,2);
\coordinate (U2) at (3.25,2.5);
\coordinate (U3) at (4.5,2);
\coordinate (U4) at (5.5,1.5);
\coordinate (U4-) at (5.13,1.5);
\coordinate (U5) at (6.5,2);
\coordinate (U6) at (7.5,2.8);
\coordinate (U7) at (8.5,2);
\coordinate (U8) at (9.75,1);
\coordinate (U9) at (10.00,2);

\coordinate (L1) at (5.00,-1);
\coordinate (L2) at (6.25,-.5);
\coordinate (L3) at (7.5,-1);
\coordinate (L4) at (8.5,-1.5);
\coordinate (L4-) at (8.13,-1.5);
\coordinate (L5) at (9.5,-1);
\coordinate (L6) at (10.5,-0.2);
\coordinate (L7) at (11.5,-1);
\coordinate (L8) at (12.75,-2);
\coordinate (L9) at (13.00,-1);

\path[fill, left color=red!5, right color=purple!15,shading angle=90]
(U1) --
plot[smooth,tension=0.8]
coordinates {
    (U1) (U3) (U4) (U6) (U9)
} --
(L9)--
(V)--
(B)--
(A)--
cycle;

\path[fill, left color=blue!5, right color=blue!15,shading angle=90]
(U1) --
plot[smooth,tension=0.8]
coordinates {
    (U1) (U3) (U4) (U6) (U9)
} --
(L9)--
(D)--
(C)--
(P)--
cycle;

\path[fill=red!30]
(U9) --
(B)--
(V)--
(L9)--
cycle;

\path[fill, left color=purple!15, right color=purple!50,shading angle=90]
(U1) --
plot[smooth,tension=0.8]
coordinates {
    (U1) (U3) (U4) (U6) (U9)
} --
(L9)--
plot[smooth,tension=0.8]
coordinates {
    (L9) (L6) (L4) (L3) (L1)
} --
cycle;

\draw[thick,red]
plot[smooth,tension=0.8]
coordinates {
    (L1) (L3) (L4) (L6) (L9)
};

\path[fill, right color=purple!50, left color=purple!0,shading angle=110]
(U4-)--(U4)--(L4)--(L4-)--
cycle;

\draw[thick, red]
plot[smooth,tension=0.8]
coordinates {
    (U1) (U3) (U4) (U6) (U9)
};

\draw[red] (U1) -- (L1);
\draw[red] (U9) -- (L9);
\draw[dashed,red] (U4-) -- (L4-);



\draw[line width=.2cm, color=white] (D) -- (9,-1);
\draw[line width=.2cm, color=red!10] (Q) -- (9,-1);
\draw[line width=.2cm, color=purple!10] (Q) -- (3.9,0);
\draw[line width=.2cm, color=white] (P) -- (4,0);
\draw[line width=.2cm, color=red!10] (B) -- (Q);
\draw (P) -- (A);
\draw (A) -- (B);
\draw (B) -- (Q);
\draw (P) -- (Q);

\draw (P) -- (C);
\draw (C) -- (D);
\draw (D) -- (Q);

\draw[-{Stealth[length=4mm]}] (P-) -- (P--);
\draw[-{Stealth[length=4mm]}] (A+) -- (A++);
\draw[-{Stealth[length=4mm]}] (C-) -- (C--);
\draw[-{Stealth[length=4mm]},red] (P+) -- (U1+);

\node[left] at (P--) {$x_\parallel$};
\node[above] at (A++) {$x^+$};
\node[below] at (C--) {$x^-$};
\node[above,red] at ($(U1++)+(.5,0)$) {$r_\perp^+ f(x_\parallel)$};
\node[above,blue] at (3,-1.5) {$W(0_\perp)$};
\node[below,red] at (V) {$W_f(0_\perp,r_\perp)$};
\end{tikzpicture}
  \caption{The deformed wedge $W_f(0_\perp, r_\perp)$ \eqref{Wdef} is in red, the original wedge $W(0_\perp)$ in blue. The edge of the shifted wedge is parameterized by a function $f(x_\parallel)$.
  In this figure, $r_\perp^+=1, r_\perp^-=0$, i.e., the wedge is shifted along the upper horizon. In the general case, the red wedge can be strictly inside the blue wedge.}
    \label{fig1}
\end{figure}
We consider QFTs on $d$-dimensional Minkowski spacetime with vacuum state $\Omega$ 
satisfying a sufficiently strong version of the Araki-Kastler axioms \cite{haag_2} such 
that the conclusions of the Bisognano-Wichmann theorem holds for wedges. It follows that, for any non-zero $f$, 
\ben
\label{hsmid}
\begin{split}
\Big(\A(W_f(x_\perp, r_\perp) \,\,  \subset\,\,  \A(W_f(x_\perp, 0)), \Omega \Big) & \quad \text{for $r_\perp^+=1,r_\perp^-=0$},\\ 
\textcolor{black}{\Big(\A(W_f(x_\perp, r_\perp))' \subset \A(W_f(x_\perp,0))', \Omega \Big)} & \quad \text{for $r_\perp^+=0, r_\perp^-=1$}
\end{split}
\een
define non-trivial HSMIs, 
with associated self-adjoint positive generators $P_+(x_\perp)$ respectively $P_-(x_\perp)$. These depend on $f$, but we will suppress this in our notation.
\textcolor{black}{We also write $K(x_\perp)$ for the generator of boosts leaving the wedge $W_f(x_\perp,0) \equiv W(x_\perp)$ (see \eqref{Wdef0}) fixed, so that 
the modular operator $\Delta(x_\perp) := \Delta_{\Omega}$ for the wedge-algebra $\A(W(x_\perp))$ satisfies $\Delta(x_\perp)^{it} = e^{-i2\pi t K(x_\perp)}$ by 
the Bisognano-Wichmann theorem, see e.g., \cite{haag_2}.}

\textcolor{black}{We have the following lemma.
\begin{lemma}
\label{lemaa:0}
Let $r_\perp$ be such that $r^-_\perp=0$ and $r^+_\perp>0$. Then 
\[
\mathcal{A}(W_f(x_\perp, r_\perp)) = e^{i r^+_\perp P_+(x_\perp)} \mathcal{A}(W_f(x_\perp,0)) e^{-i r^+_\perp P_+(x_\perp)},
\]
i.e., $P_+(x_\perp)$ has a geometric action. Similarly, if $r^+_\perp=0$ and $r^-_\perp>0$, then 
\[
\mathcal{A}(W_f(x_\perp, r_\perp))' = e^{i r^-_\perp P_-(x_\perp)} \mathcal{A}(W_f(x_\perp,0))' e^{-i r^-_\perp P_-(x_\perp)}.
\]
\end{lemma}
\begin{proof}
We give the proof for the "$+$" only since the other case is analogous by properties of the HSMIs for the commutants. 
First since the action of $\Delta(x_\perp)^{it}$ is geometric, we have 
\ben
\Delta(x_\perp)^{-it}\mathcal{A}(W_f(x_\perp, r_\perp))\Delta(x_\perp)^{it} = \mathcal{A}(W_f(x_\perp, e^{2\pi t}r_\perp)).
\een
Now we take $t$ such that $e^{2\pi t} = (r^+_\perp)^{-1}$. Then, by the property 4) of the HSMIs and the first line of \eqref{hsmid}, 
\ben
\mathcal{A}(W_f(x_\perp, e^{2\pi t}r_\perp))=e^{iP_+(x_\perp)} \mathcal{A}(W_f(x_\perp, 0)) e^{-iP_+(x_\perp)}.
\een
By combining these two relations and using the property 6) of HSMIs, we get the statement.
\end{proof}
}

As in the 2-dimensional case, we consider the future directed timelike or null vector $r_\perp$ and its dual \eqref{dual} $\tilde r_\perp$, which in our conventions is a spacelike or null 
vector pointing to the right. Then we first define the variational derivative
\ben
\label{vardefS}
r_\perp^\mu \delta_\mu S(x_\perp) :=  \frac{d}{d\lambda} S(\Phi |\! | \Omega)_{\A(W_f(x_\perp, \lambda r_\perp))} \Bigg|_{\lambda=0}, 
\een
which implicitly depends on $f$. Then \textcolor{black}{using lemma \ref{lemaa:0}} we have a formula analogous to \eqref{antboth}
\ben
\label{antboth1}
\begin{split}
-r_\perp^\mu \delta_\mu S(x_\perp)
=& 2\pi \inf \Big\{ \tilde r^+_\perp (u_+' \, \Phi, P_+(x_\perp) \, u_+'\Phi) + \tilde r^-_\perp (u_-' \, \Phi, P_-(x_\perp) \, u_-'\Phi) \, :\\ 
 &\hspace{1cm} u_\pm' \in \A(W_f(x_{\perp},0))', \, u'_\pm \, {\rm unitary} \Big\} ,
\end{split}
\een
under analogous domain conditions on $\Phi$.
Contrary to the 2-dimensional case, this variational formula here does not imply that $-r_\perp^\mu \delta_\mu S(x_\perp)\ge -r_\perp^\mu \delta_\mu S(x_\perp+s_\perp)$
when $s_\perp$ is a spacelike or null vector that is pointing to the right, because, contrary to the 2-dimensional case, $P_\pm(x_\perp)$ may depend on $x_\perp$ due to the deformation 
induced by $f$. 

To address this issue, we use the explicit formula for the minimizer in the ant formulas \eqref{antleft} and \eqref{antright} with the Connes cocycle 
$u'_t$ for $\A(W_f(x_{\perp},0))'$ with respect to $\Phi, \Omega$, giving
\ben
\label{antboth2}
-r_\perp^\mu \delta_\mu S(x_\perp)
= 2\pi \tilde r^+_\perp \lim_{t \to +\infty}(u_t' \, \Phi, P_+(x_\perp) \, u_t'\Phi) + 2\pi \tilde r^-_\perp \lim_{t \to -\infty} (u_t' \, \Phi, P_-(x_\perp) \, u_t'\Phi) \, 
\een
Using the ant formulas again then gives
\ben
\label{antboth3}
\begin{split}
-r_\perp^\mu \delta_\mu S(x_\perp)
\ge & \,  2\pi \tilde r^+_\perp \lim_{t \to +\infty}\left( u_t' \Phi, \Big[ P_+(x_\perp) - P_+(x_\perp+s_\perp) \Big] u_t'\Phi \right) +\\
& \, 2\pi \tilde r^-_\perp \lim_{t \to -\infty} \left(u_t'  \Phi, \Big[ P_-(x_\perp) - P_-(x_\perp+s_\perp) \Big] u_t'\Phi \right) +\\
& \, 2\pi \inf \Big\{ \tilde r^+_\perp (u_+' \, \Phi, P_+(x_\perp+s_\perp) \, u_+'\Phi) + \tilde r^-_\perp (u_-' \, \Phi, P_-(x_\perp+s_\perp) \, u_-'\Phi) \, :\\ 
 &\hspace{1cm} u_\pm' \in \A(W_f(x_{\perp}+s_\perp,0))', \, u'_\pm \, {\rm unitary} \Big\} ,
\end{split}
\een
which using the ant formulas one more time gives:
\begin{theorem}
\label{thm6.1}
Under the domain assumptions analogous to proposition \ref{propant} on $\Phi$, we have 
\ben
\label{antboth4}
\begin{split}
&-\Big[ r_\perp^\mu \delta_\mu S(x_\perp) - r_\perp^\mu \delta_\mu S(x_\perp+s_\perp) \Big] \\
\ge & \,  2\pi \tilde r^+_\perp \lim_{t \to +\infty}\left( u_t' \Phi, \Big[ P_+(x_\perp) - P_+(x_\perp+s_\perp) \Big] u_t'\Phi \right) +\\
& \, 2\pi \tilde r^-_\perp \lim_{t \to -\infty} \left(u_t'  \Phi, \Big[ P_-(x_\perp) - P_-(x_\perp+s_\perp) \Big] u_t'\Phi \right) \, .
\end{split}
\een
for any future pointing timelike or null vector $r_\perp$ with dual \eqref{dual} $\tilde r_\perp$, and any spacelike or null vector $s_\perp$ pointing to the right. 
The definition of $P_\pm(x_\perp)$ and of the variational derivative \eqref{vardefS} $\delta S(x_\perp)$ implicitly depend on the chosen testfunction $f(x_\parallel) \ge 0$.
\end{theorem}

\subsection{Other forms of the QDEC}

\noindent
{\bf Relation between $P_\pm(x_\perp)$ and the stress energy tensor for free scalar field theory of mass $m>0$.} 
It has been argued heuristically \cite{Markov} that the generators of the HSMIs \eqref{hsmid}
should be given in terms of the quantum stress energy tensor by\footnote{Note that, assuming \eqref{PT}, $P_\pm(x_\perp)$ only depends on $x^\mp$.}
\ben
\label{PT}
P_\pm(x_\perp) = \int T_{\pm \pm}(x) f(x_\parallel) \, dx^\pm dx_\parallel \quad \text{(formal)}
\een
This formula should be considered as formal, e.g., because one expects an additional contribution from future and past null-infinity 
if the theory contains any massless degrees of freedom, associated with stress energy tensor integrals along the 
red portions of $\sI^\pm$ in fig. \ref{fig:nullplane}, see e.g., \cite{speranza} for an in-depth discussion. 
\begin{figure}
\centering
\begin{tikzpicture}[scale=3]
\draw[double distance=2pt] (0,1) -- (1,0) -- (0,-1) -- (-1,0) -- cycle;
\node at (0,1.12) {$i^+$};
\node[left,color=red] at (-1.02,0) {$i^+$};
\node at (0,-1.12) {$i^-$};
\node at (-.78,-.6) {${{\mathscr I}^-}$};
\node at (.68,.7) {${{\mathscr I}^+}$};
\draw[red, very thick] (-.25,-.75) -- (.75,.25);
\draw[red, very thick] (-.25,-.75) -- (-1,0);
\draw[red, very thick] (1,0) -- (.75,.25);
\draw[blue, very thick, dashed] (.87,-.07) -- (.65,.15);
\node[red,right] at (0.05,-0.5) {$\{x^-=x^-_0\}$};
\draw[color=white, line width=4mm] (0.77,.1)  -- (1.77,.1);
\node[blue,right] at (0.77,.1) {$\{ x^+ = x^+_0, x^- \le x^-_0\}$};
\draw[dashed] (0,-1) -- (0,1);
\fill (0,1) circle (1pt);
\fill (0,-1) circle (1pt);
\fill[color=red] (1,0) circle (1pt);
\fill[color=red] (-1,0) circle (1pt);
\fill[color=blue] (.9,-.1) circle (1pt);
\end{tikzpicture}
\caption{The pieces of $P_+(x_\perp)$ associated with the null-infinities in a 2-dimensional Penrose diagram of Minkowski spacetime are represented by the 
red line segments in ${\mathscr I}^\pm$ .}
\label{fig:nullplane}
\end{figure}
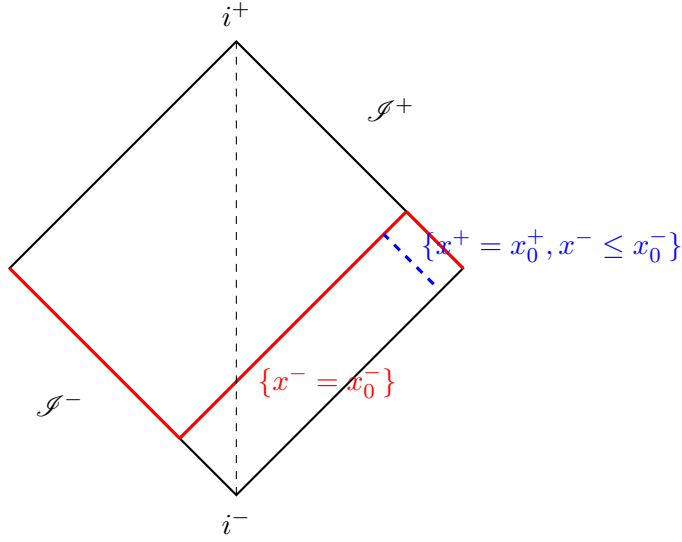
They would presumably not lead to a different final result \eqref{+-QDEC}. But, as we now 
argue, these do not appear at any rate in a free scalar QFT of mass $m>0$, and the formal expression \eqref{PT}
can be given a rigorous meaning in this setting. 

It is well-known that the algebraic construction of this theory leads to a representation on the bosonic Fock space
$\H$ over $\h$, which is called the 1-particle subspace, see e.g., Sec. \ref{coherent} and 
\cite{sanders,wald2}. Concretely, $\h = L^2(\RR^{d-1}, dp/(2E_p))$, which can be 
regarded as the square integrable functions on the upper mass hyperboloid with respect to the invariant measure. 
The vacuum state $\Omega$ is the normalized state with zero particles in $\H$.

As is well-known \cite{Brunetti:2002nt}, the concept of real linear standard subspace $H \subset \h$ can be used to define modular operators at the 1-particle level: 
There is a self-adjoint positive operator $\delta_H$ on $\h$ and an anti-unitary $j_H$ such that $j_H H =H'=\{ \Phi \in \h : 
\sigma(\Phi, \Psi) \equiv 2{\rm Im} (\Phi, \Psi) = 0 \, \forall \Psi \in H\}$, and such that $\delta_H^{it} H = H$ for all $t \in \RR$.
If $O \subset \RR^d$ is 
the causal completion of an open proper subset of the time zero plane with regular boundary, a standard subspace $H(O)$ can be defined as $H(O):= \{ \Phi = KF : F \in C_0^\infty(O, \RR)\}$, 
where $(KF)(p) := \hat F(E_p, p)$. Then it turns out \cite{LRT78} (see also \cite{Brunetti:2002nt}) 
that the modular flow $\Delta_{\Omega, \A(O)}$ is the second quantization 
of $\delta^{it}_{H(O)}$, i.e. $\Delta_{\Omega, \A(O)}^{it} = \oplus_{n \ge 0} (\delta^{it}_{H(O)})^{\otimes n}$. In particular $H(O)$ is a standard subspace e.g. for the wedges \eqref{Wdef} $O=W_f(x_\perp,0)$, 
and the theory applies. 

Let $\h_{U(1)}$ be the 1-particle Hilbert space of the chiral $U(1)$ current  QFT. This theory defines a net 
of standard subspaces $I \mapsto H_{U(1)}(I)\subset \h_{U(1)}$ for every open non-empty (possibly semi-infinite) interval $I \subset \RR$
\cite[Sec.\ 2.4
]{Morinelli:2021nsx}. By a slight extension of the results of \cite[Sec. 3]{Morinelli:2021nsx} one can show that the decomposition of the restriction of the Klein-Gordon field to the null plane $\{ x^- = x^-_0\}$ into an integral of $U(1)$ currents living on the null lines at fixed $x_\parallel$ gives rise to a unitary
\ben
\label{Vdef}
V(x^-_0) : \h \to \int\limits^\oplus_{\RR^{d-2}} \h_{U(1)}(x_\parallel) \, dx_\parallel \, ,
\een
with $\h_{U(1)}(x_\parallel) = \h_{U(1)}$ a constant field of Hilbert spaces. More precisely, for $F$ a real distribution of the form $F(x^+, x^-, x_\parallel) = \delta(x^- - x^-_0) \partial_+F_+(x^+) F_\parallel(x_\parallel)$ with $F_+$ smooth and compactly supported in a interval $I \subset \RR$ and  $F_\parallel$ smooth and compactly supported in an open set $I_\parallel \subset \RR^{d-2}$, one obtains that $KF \in \h$ and $V(x_0^-)KF \in \int^\oplus_{I_\parallel} H_{U(1)}(I) dx_\parallel$.

As a consequence, if we consider the deformed wedge $W_f(x_\perp, r_\perp)$ with $r_\perp^+ = 1$, $r_\perp^- = 0$, and we choose $x_0^- = x_\perp^-$ corresponding to the upper horizon of $W_f(x_\perp, r_\perp)$, we obtain, analogously to~\cite[Prop.~3.11,~3.13]{Morinelli:2021nsx}\footnote{
The condition that $m>0$ is used implicitly in~\cite[Lem.~3.6,~3.7]{Morinelli:2021nsx} on which these propostions build. 
},
\ben
\label{Vdef2}
V(x_\perp^-) H(W_f(x_\perp, r_\perp)) = \int\limits^\oplus_{\RR^{d-2}} H_{U(1)}(x_\perp^+ +  f(x_\parallel), \infty) \, dx_\parallel \, .
\een
In the standard subspace formalism, there is an analogue of HSMIs, see e.g., \cite[Sec. 3.2]{Morinelli:2021nsx}.
Let $P_{U(1)}(x_\parallel)$ be the generator of translations on associated with the 
HSMI of standard subspaces $H_{U(1)}(0,\infty) \supset H_{U(1)}(f(x_\parallel), \infty)$. Each $P_{U(1)}(x_\parallel)$
is a non-negative self-adjoint operator on $\h_{U(1)}(x_\parallel)$, an identical copy of $\h_{U(1)}$. Then we have 
\ben
\label{PT1}
P_+(x_\perp) \Big|_{\h} =  V(x^-_\perp)^* \left( \, \, \, \int\limits^\oplus_{\RR^{d-2}} P_{U(1)}(x_\parallel) \, dx_\parallel \right) V(x^-_\perp)
\een
as a consequence of \cite[Thm. 4.3]{Morinelli:2021nsx} and \eqref{Vdef2}. 

$P_{U(1)}$ can be related to 
the stress energy tensor of the $U(1)$-current net as follows. On the full, second-quantized, Hilbert-spaces $\H_{U(1)}:=\oplus_{n=0}^\infty
\otimes_{\rm sym}^n \h_{U(1)}$, the chiral $U(1)$ current $j_{U(1)}(u)$ is a Wightman field on the real line with the commutation relation 
$[j_{U(1)}(u), j_{U(1)}(u')] = i\delta'(u-u')1$ and
associated normal ordered energy momentum tensor $T_{U(1)}(u) = \, : j_{U(1)}(u)^2 :$. Then, since 
$P_{U(1)}(x_\parallel)$ implements a translation by $f(x_\parallel)$ by the 1-particle version of Borchers' theorem \cite{Bor1,Bor2},
\ben
P_{U(1)}(x_\parallel) = f(x_\parallel) \int\limits_{-\infty}^\infty T_{U(1)}(x^+) \, dx^+ \Big|_{\h_{U(1)}}.
\een
Therefore, since it is easily seen that $P_\pm(x_\perp)$ is of second quantized form, \eqref{PT1} is telling us that \eqref{PT} holds (for $++$), provided that 
we define the null-smeared stress energy tensor for $g \in C^\infty_0(\RR^{d-1})$ as (see also rem. \ref{ctpr})
\ben
\label{PT2}
T_{++}(x^-,g) 
:=  \Gamma(x^-)^* d\Gamma \left( \, \, \, \int\limits^\oplus_{\RR^{d-2}} \int\limits_{-\infty}^\infty  g(x^+, x_\parallel) \, T_{U(1)}(x^+) \, dx^+ \Big|_{\h_{U(1)}} \, dx_\parallel \right) \Gamma(x^-), 
\een
where $\Gamma(x^-) = \Gamma[V(x^-)]$ is the second quantization ($\Gamma, d\Gamma$ are the second quantization functors \cite[Sec. X.7]{RS}), 
and where the integral is understood in the weak sense.
By a reflection, an analogous result holds for $--$.

In combination with Lem. \ref{dual}, Thm. \ref{thm6.1}, and \eqref{PT} we get:
\begin{corollary}
\label{thm6.1}
Consider a free KG QFT of mass $m > 0$ in $d>2$ spacetime dimensions.
Under the domain assumptions analogous to Prop. \ref{propant} on $\Phi$ [where $s=(s_\perp, 0_\parallel)$], 
\ben
\label{antboth5}
\begin{split}
&-\Big[ r_\perp^\mu \delta_\mu S(x_\perp) - r_\perp^\mu \delta_\mu S(x_\perp+s_\perp) \Big] \\
\ge & \,  2\pi \tilde r^+_\perp \lim_{t \to +\infty}
\int\limits_{\RR^{d-1}} f(x_\parallel) \left( u_t' \Phi, \Big[ T_{++} (x) - T_{++}(x+s) \Big] u_t'\Phi \right) \, dx^+ dx_\parallel +\\
& \, 2\pi \tilde r^-_\perp \lim_{t \to -\infty} \int\limits_{\RR^{d-1}} f(x_\parallel) 
\left(u_t'  \Phi, \Big[ T_{--} (x) - T_{--}(x+s)  \Big] u_t'\Phi \right) \, \, dx^- dx_\parallel  ,
\end{split}
\een
for any future pointing timelike or null vector $r_\perp$ with dual \eqref{dual} $\tilde r_\perp$, and any right-pointing 
spacelike or null $s_\perp$.
The definition of the variational derivative \eqref{vardefS} $\delta S(x_\perp)$ implicitly depend on the chosen testfunction $f(x_\parallel) \ge 0$.
\end{corollary}
\begin{remark}
An argument for \eqref{PT} based on a characterization of the stress energy tensor in the generally covariant approach to (not necessarily free) QFT  \cite{Hollands:2004yh,BFV} 
will be given in a forthcoming work \cite{HM2}.
\end{remark}

\begin{remark}
\label{ctpr}
The null-smeared stress energy tensor \eqref{PT2} could be written as
$\int_{\RR^{d-2} }dx_\parallel A(x_\parallel)$
where, for each $x_\parallel$, $A(x_\parallel)$ is an operator on the continuous tensor 
product \cite{Nap} of the constant field of Fock spaces of the U(1) current. 

We also note that from~\eqref{PT1}, by a straightfoward computation one shows that in the sense of quadratic forms on a suitable dense domain
\[
\int\limits_{-\infty}^{+\infty}  g(x^-) P_+(x_\perp) \,dx^-= \lim_{n \to +\infty} \int\limits_{\RR^d} h_n(x^+)g(x^-)f(x_\parallel):T_{++}:(x)\,dx
\]
with $g, h_n$ smooth compactly supported functions on $\RR$ such that $h_n$ converges to the unit function, and with $:T_{++}:$ the normal ordered stress energy tensor of the KG-field.
 
%
\end{remark}

\medskip
\noindent
{\bf Heuristic form of \eqref{antboth4} in terms of entanglement entropy.}
Taking the derivative in $s_\perp$ in \eqref{antboth5}, i.e., dividing by $-s^->0$ and letting $s^- \to 0$, and assuming that it commutes with the limits on the right side, we obtain, for example 
\ben
-\partial_- \delta_+ S(x_\perp) \ge 2\pi  \lim_{t \to +\infty} \int \Big(u_t' \Phi, \partial_- T_{++}(x) \, u_t' \Phi \Big) \, f(x_\parallel) \, dx^+dx_\parallel . 
\een
Using now the conservation of the stress energy tensor, $\partial_+ T_{-+} + \partial_- T_{++} - \sum_{j=2}^{d-1} \partial_j T_{j+} = 0$, and performing an integration by parts
with respect to the derivatives $\partial_j$, we get
\ben
-\partial_- \delta_+ S(x_\perp) \ge -2\pi  \lim_{t \to +\infty} \int \sum_{j=2}^{d-1} \Big(u_t' \Phi,  T_{j+}(x) \, u_t' \Phi \Big) \, \partial_j f(x_\parallel) \, dx^+ dx_\parallel . 
\een
We now formally split the $dx^+$ integral into an integral over $(-\infty, x^+_\perp)$, which gives the expectation value of an operator $M'$ formally affiliated 
with $\A(W_f(x_\perp, 0))'$, and an integral over $(x^+_\perp, +\infty)$, which gives the expectation value of an operator $M$ formally affiliated 
with $\A(W_f(x_\perp, 0))$, formally commuting with the unitary $u'_t$. 
For the element $M' \in \A(W_f(x_\perp, 0))'$, one may guess from the ergodic nature of the modular flow \eqref{ergodic} that this contribution 
tends to zero as $t \to \infty$. Hence, we have given evidence for:
\begin{conjecture}
For a class of sufficiently regular states $\Phi$, we have
\ben
\label{QDECnonrig}
-\partial_- \delta_+ S(x_\perp) \ge -2\pi  \int\limits_{{\mathbb R}^{d-2}} \int\limits_{x_\perp^+}^\infty \sum_{j=2}^{d-1} \Big( \Phi, T_{j+}(x) \, \Phi\Big) \, \partial_j f(x_\parallel) \, dx^+ dx_\parallel . 
\een
\end{conjecture}
So far, our inequality \eqref{QDECnonrig} is still in terms of the relative entropy, and has a good chance to hold rigorously. We will not try to prove this here because one would first have to demonstrate \eqref{PT} including the subtle question of the domain and precise meaning of the operator-valued integral, which by itself is not trivial. 

Instead, we now decompose the relative entropy into the entanglement entropy, $S_{\rm EE}$, and expectation value of the energy momentum flux operator, i.e., we formally write, 
for $r^-=0, r^+>0$, (see e.g., \cite{Markov,GGI})
\ben
\begin{split}
&S(\Phi |\! | \Omega)_{\A(W_f(x_\perp, r_\perp))} = \, - S_{\rm EE}[W_f(x_\perp, r_\perp)] \\
& + \, 2\pi  \int\limits_{{\mathbb R}^{d-2}} \int\limits_{x^+_\perp+r^+ f(x_\parallel)}^\infty [y^+ -x^+_\perp - r^+ f(x_\parallel)] \, \Big(
\Phi, T_{++}(y^+, x^-_\perp, x_\parallel)  \, \Phi \Big) \, dy^+ dx_\parallel \, , 
\end{split}
\een
where $S_{\rm EE}$ is the entanglement entropy of the reduced state (``reduced density matrix'') of $\Phi$ with respect to the deformed wedge $W_f(x_\perp, r_\perp)$ \eqref{Wdef}. 
This quantity is itself infinite, but is expected to yield a finite expression after taking the derivatives $-\partial_- \delta_+$ [see \eqref{vardefS} for the definition of $\delta_+$], because these should remove the infinite part, see Sec. \ref{CST} for further discussion. A small computation using \eqref{vardefS} and the 
conservation of stress energy similar to that above shows that, 
after these derivatives are taken, one obtains a term precisely equal to the right side of \eqref{QDECnonrig}. Carrying out this cancellation, and formally taking $f$ to be a 
delta-function, one finds
\ben
\label{+-QDEC}
2\pi (\Phi, T_{+-}(x) \Phi) \ge -\partial_- \frac{\delta}{\delta C^+(x)} S_{\rm EE}[C], 
\een
which is our final form of the QDEC for a special choice of indices. 

Here, $C$ is understood as the edge of the wedge $W(x_\perp)$, and 
the variational derivative \eqref{vardefS} has been rewritten as a shape variation defined as follows.
Let $F[C]$ be a functional of spacelike, smooth, codimenision-2 embedded surfaces $C$ in a spacetime $(M, g)$, and let $X = X^\mu \partial_\mu$ 
be a vector field defined on $C$ and transversal to $C$. We extend $X$ arbitrarily off of $C$, and let $\varphi_t$ be the flow of $X$. Then, assuming the derivative exists, we write
\ben
\label{vardef}
\frac{d}{dt} F[\varphi^*_t C]  \bigg|_{t=0} = \int\limits_C X^\mu(x) \frac{\delta}{\delta C^\mu(x)} F \, dA(x) \equiv \delta_X F,
\een
where $dA=\nu_{\mu\nu}dx^\mu \wedge dx^\nu$ is the positively oriented\footnote{
\label{orientationfootnote} We orient $C$ in a given but fixed way. Our orientations for 
Minkowski spacetime (or more general spacetime) are then such that $\epsilon = l \wedge k \wedge \nu$ is positively oriented 
where $k$ is a future directed, outward pointing (relative to the volume enclosed by $C$) null vector, and $-l$ is a future directed, inward pointing null vector defined near $C$.} 
volume form on $C$ induced by the metric $g$. Our variational derivative \eqref{vardefS} then corresponds to taking $X^\mu = fr^\mu$.
 
The choices ``$--$'' and ``$++$'' in \eqref{antboth4} simply give the QNEC, and combined with \eqref{+-QDEC} gives the form of QDEC
stated in \eqref{qdec} in the Introduction. In particular, the restriction that $u$ in \eqref{qdec} be orthogonal to $C$ follows and 
may also be stated as $u \wedge k = \lambda n$ for some $\lambda$, where $n$ is the binormal to $C$.

\section{Extensions of the QDEC}
\label{extensions}

Here we discuss several potential extensions and consequences of the QDEC. 

\subsection{More general cuts or curved spacetimes}
\label{CST}
So far, we have considered entangling cuts $C$ given by the edge of a Rindler wedge in Minkowski spacetime.
It is natural to ask to what extent the QDEC could be generalized to cuts $C = \partial A$ bounding some more general domain $A$ 
within some Cauchy surface $\Sigma \supset A$ of Minkowski spacetime, or even of a more general globally hyperbolic curved spacetime. 

Given such a cut, we may define two future directed null vectors $k=k^\mu \partial_\mu, l=l^\mu \partial_\mu$ orthogonal to $C$ such that 
$l^\mu k_\mu = -1$ and such that the projections of $l,k$ to $\Sigma$ are inward- respectively outward pointing on $C=\partial A$. We can 
locally consider the outgoing future directed null sheet $H^+$ generated by affinely parameterized geodesics whose initial tangent on $C$ is $k$, 
which is a subset of $\partial J^+(A) \setminus A$. Likewise, we can locally consider the outgoing past 
directed null sheet $H^-$ generated by affinely parameterized geodesics whose initial tangent on $C$ is $-l$, which is a subset of $\partial J^-(A) \setminus A$. 

The question is what geometric properties these null sheets might need to possess in order that the QDEC \eqref{qdec} still holds
for all future directed causal $u = u^\mu \partial_\mu$ orthogonal to $C$, 
where as before $\langle T_{\mu\nu}\rangle = (\Phi, T_{\mu\nu} \Phi)$, where $n^{\alpha\beta} = 2l^{[\alpha} k^{\beta]}$ is the bi-normal to $C$, 
and where the variational derivative is defined by \eqref{vardef}. 

To discuss this, it is useful to introduce the null expansions $\theta_k$ and $\theta_l$. 
If $\phi, \psi$ are smooth functions on $C$, we can consider the variational derivative $\delta_X \theta_k$ or $\delta_X \theta_l$ \eqref{vardef} for 
$X=\psi l + \phi k$. It is well-known \cite{mars} that, e.g., 
\ben
\delta_{\psi l + \phi k} \theta_k = W\phi + L\psi,
\een
where 
\begin{subequations}
\ben
\label{Ldef}
L\psi := -q^{\alpha\beta} D_\alpha D_\beta \psi + \beta^\mu D_\mu \psi + \left[ \theta_k \theta_l -G_{\mu\nu} k^\mu l^\nu + \frac{1}{2}\left( R[q] -\frac{1}{2} \beta^\mu \beta_\mu + D_\mu \beta^\mu \right) \right]\psi, 
\een
\ben
\label{Wphi}
W\phi := (-\theta_k^2 - \sigma_k^2 - R_{\mu\nu} k^\mu k^\nu) \phi . 
\een
\end{subequations}
Here, $\beta_\mu dx^\mu = 2l_\nu \nabla^\alpha k^\nu q_{\alpha\mu} \, dx^\mu$ 
is the natural connection 1-form on the normal bundle $TC^\perp$, $q^{\alpha\beta}=g^{\alpha\beta} + 2l^{(\alpha} k^{\beta)}$ 
is the contravariant induced metric on $C$, where $D_\mu$ is the intrinsic derivative operator compatible with $q^{\alpha\beta}$,
and where $\sigma_k^{\mu\nu}$ is the shear of $k$.
The elliptic operator $L$ is called the stability operator. Even though it is not self adjoint on $L^2(C,dA)$, it is known \cite[Lem. 4.1]{mars} that for compact $C$, 
the "principal eigenvalue" $\lambda_1$, i.e. that with the smallest real part, is actually real, and that the corresponding eigenfunction $\psi_1>0$ everywhere on $C$.

After these preliminaries, we now consider the variational derivatives \eqref{vardef} of the entanglement entropy appearing in the QDEC \eqref{qdec} for 
a compact $C$. It is well-known that the entanglement entropy is actually infinite, i.e., it must be defined with 
some regulator, e.g., by replacing $S_{\rm EE}[C]$ by $\frac{1}{2}I[A:B]$, where $B$ is the domain within $\Sigma$ such that $\Sigma = A \cup B \cup N$
(disjoint union), where $N \cong C \times [0,\ell]$ is a 1-sided normal neighborhood of $C$ of width $\ell$, see Fig. \ref{fig:aibi1} for an illustration.

\begin{figure}[h!]
\begin{center}
\begin{tikzpicture}[square/.style={regular polygon,regular polygon sides=4}]
            \draw[thick,blue] ([shift=(60:1cm)] 0, 0) arc (60:120:3cm);
            \draw[thick,red] ([shift=(60:1.2cm)] 0, 0) arc (60:120:3.2cm);

        \filldraw[fill=blue!20!white, draw=blue!50]
        ([shift=(60:1cm)] 0, 0) arc (60:120:3cm) -- (-1.5,0) -- (-0.5, 0) -- cycle;
        \filldraw[fill=red!20!white, draw=red!50]
        ([shift=(60:1.2cm)] 0, 0) arc (60:120:3.2cm) -- (-2.0,2.5) -- (0.0, 2.5) -- cycle;




        \node at (-2, 1.21) {$\ell$};
        \node at (-2.5, 2.25) {\color{red!50}{$B$}};
        \node at (-2.5, 0.4) {\color{blue!50}{$A$}};
\end{tikzpicture}

\end{center}
\caption{The the sets $A,B$ separated by $\ell$.}
\label{fig:aibi1}
\end{figure}

$I$ is the mutual information associated with the state $\Phi$, 
which is well-defined in such a setting due to the split property, see e.g., \cite{sanders}. Then, under this replacement, one has
\ben
S_{\rm EE}[C] := \frac{1}{2}I[A:B] = S_{\rm EE}^{\rm div}[C] + S_{\rm EE}^{\rm fin}[C], 
\een
where the (non-unique) split is into a divergent part as $\ell \to 0$, and a finite part. One
expects, e.g. in $d=4$ dimensions, that, at least for a conformally invariant theory (CFT) \cite{casini}, \cite{Solodukhin:2008dh}, 
\ben
\label{Sdiv}
S_{\rm EE}^{\rm div}[C] = \frac{c_2 \, {\rm Area}[C]}{\ell^2}  +  \log \ell \left\{ c_0^{\rm I} \, \chi[C] + c_0^{\rm II} \!
\int\limits_C \left[
2\left( R_{\alpha\beta} - \frac{1}{3}Rg_{\alpha\beta} \right)l^\alpha k^\beta  -  \theta_k \theta_l
\right] dA \right\}
\een
 where $c_2, c_0^{\rm I,II}$ are theory-dependent constants\footnote{$c_0^{\rm I,II}$ can be expressed in terms of the $a$ and $c$ trace-anomaly coefficients
defined by $T_\mu^\mu = c W_{\mu\nu\sigma\rho}W^{\mu\nu\sigma\rho} + a E_4$, where $E_4$ is the Euler density in $d=4$
and $W_{\mu\nu\sigma\rho}$ is the Weyl tensor \cite{Solodukhin:2008dh}.}, 
and where $\chi[C] = (2\pi)^{-1} \int_C R[q] \, dA$ is the Euler characteristic of $C$. In order for the QDEC \eqref{qdec} to make any sense in 
this more general situation, the derivatives in \eqref{qdec} must cancel $S_{\rm EE}^{\rm div}$. The next proposition gives sufficient geometric conditions for this to happen. 

\begin{proposition}
\label{label:geomprop}
Suppose that $C=\partial A$, $A \subset \Sigma$, where $\Sigma$ is an achronal surface and $A$ is compact with smooth boundary in a $4$-dimensional spacetime. 
Suppose that
\begin{itemize}
\item $\theta_k=0=\sigma_k^{\mu\nu}$ on $C$,
\item $G_{\mu\nu} = -\Lambda g_{\mu\nu}$ in some neighborhood of $C$, 
\item The principal eigenvalue of $L$ vanishes, $\lambda_1=0$. 
\end{itemize}
Then on $C$, we have $u^\alpha \nabla_\alpha \, (k^\beta \delta S_{\rm EE}^{\rm div}[C]/\delta C^\beta) = 0$ for any $u$ orthogonal to $C$.
\end{proposition}

We defer the proof to the end of this section. Our derivation of the QDEC, valid for the case when $C$ is the edge of a Rindler wedge in Minkowski spacetime, made essential use of global symmetries and special geometric features. 
In particular both future and past horizons $H^\pm$ of the Rindler wedge have the property that $k$ respectively $l$ are expansion- and shear free, 
affinely parameterized vector fields tangent to $H^\pm$, respectively. Our derivation proceed by deforming the edge of the wedge, $C$, to a wiggly cut along $k$ (tangent to $H^+$). Then, we considered a Rindler wedge slightly inside and repeated this construction. This (infinitesimally) shifted Rindler wedge again has vanishing expansion along $k$. It is plausible that, in a general case, we would require at least these geometric features, i.e., (i) an outgoing null sheet with expansion free tangent and normal $k$ (i.e., $k$ is tangent to the geodesics ruling the null sheet), such that (ii) there is an infinitesimally nearby expansion free null-sheet. 

For requirement (i), we thus need $\theta_k=0$, at least at $C$. 
Requirement  (ii) is fulfilled, at least at $C$, if and only if the varied expansion $\delta_{\psi l} \theta_k$
in the direction $\psi l$ off of $C$ vanishes, where $\psi>0$ is some function on $C$. This is ensured by $\lambda_1=0$ \cite[Lem. 4.1]{mars}. 
In view also of the previous proposition, we therefore conjecture:

\begin{conjecture}
 Let $C$ be a cut in some curved spacetime such that the assumptions of Prop. \ref{label:geomprop} hold. Then the QDEC holds.
 \end{conjecture}
 
\noindent
{\bf Example.}
A horizon $H$ of a stationary black hole has $\theta_k=0=\sigma_k^{\mu\nu}$ everywhere ($k$ is tangent to the affinely parameterized null geodesic generators ruling $H$).  
On the other hand, if the horizon is {\it non-degenerate}, meaning that the horizon Killing field satisfies $K^\alpha \nabla_\alpha K^\beta = \kappa K^\beta$ for $\kappa > 0$, 
then the principal eigenvalue $\lambda_1$ of $L$ is strictly {\it positive} \cite{HollandsWald}, so that the assumption $\lambda_1=0$ of the conjecture does {\it not} hold. 
Thus, we would not necessarily expect cuts of sub-extremal Reissner-Nordström or Kerr black hole horizons to satisfy the QDEC. By contrast, we have\footnote{See also \cite{Mars1}
for properties of the stability operator in the presence of spacetime symmetry.}:

\begin{proposition}
For any degenerate  $(\kappa = 0)$, compactly generated Killing horizon $H$ in a $d$-dimensional spacetime such that $G_{\mu\nu} = -\Lambda g_{\mu\nu}$
locally near a cut $C \subset H$, the principal eigenvalue $\lambda_1$ of $L$ \eqref{Ldef} vanishes.
\end{proposition}

\begin{proof}
We take $k:=K|_{H}$, define $l$ such that $l^\mu k_\mu = -1$ on $C$, and extend $l$ by shooting of null geodesics tangent to $l$ off of $C$. 
With this, $G_{\mu\nu} = -\Lambda g_{\mu\nu}$ gives\footnote{The signs in front of $D_\mu \beta^\mu$ here and in \eqref{Ldef} do not match, which is intentional.} 
\ben
\label{Ldef1}
R[q] -\frac{1}{2} \beta^\mu \beta_\mu - D_\mu \beta^\mu -2\Lambda= 0 
\een
on $C$ using that $\theta_k=0=\sigma_k^{\mu\nu}$ and that $K$ is Killing, see e.g., \cite[App. A]{Hollands:2022fkn}\footnote{In the present setting $K = \partial_v$ in the coordinates  \cite[App. A]{Hollands:2022fkn}, and so all terms containing a "$\partial_v$"
in \cite[App. A]{Hollands:2022fkn} would be zero.}.
Combining this with the formula \eqref{Ldef} for $L$ and using again that $\theta_k=0$ gives
\ben
L = q^{\mu\nu} \nabla_\mu(-\nabla_\nu+\beta_\nu).
\een
Using the above divergence structure of $L$ and integrating 
$L\psi_1 = \lambda_1 \psi_1$ over $C$ gives $0=\lambda_1 \int_C \psi_1 \, dA$, and since $\psi_1 > 0$ \cite[Lem. 4.1]{mars}, we find that $\lambda_1=0$.
\end{proof}

Examples of spacetimes hosting degenerate Killing horizons to which the proposition applies are 
the {\it extremal} Kerr black hole horizons, or their (A)dS counterparts. The proposition can also be generalized to the Einstein-Maxwell-$\Lambda$ equations, 
so the same would be true for {\it extremal} Kerr-Newman-(A)dS spacetimes.
Thus, according to our conjecture, the QDEC should hold for any cut $C$ of the event 
horizon in these extremal black hole spacetimes.

\medskip
\noindent
{\bf Example.}
As another example where the QDEC would be expected even though not all conditions of Prop. \ref{label:geomprop} are met locally, consider two copies of Minkowski spacetime connected by a throat region, 
\ben
g = -dt^2 + dr^2 + \rho(r)^2 g_{S^2},
\een
where $t,r \in \mathbb{R}$, where $g_{S^2}$ is the round metric on a unit radius $S^2$, and where $\rho$ is a smooth positive function such that 
$\rho(r) = r$ for $|r|>2r_0$, and $\rho(r) = \rho_0$ for $|r|<r_0$. The throat region at some fixed time is defined by $|r|<r_0$ and $t=t_0$. For $|r| < r_0$ the spacetime is 
is in fact an ultrastatic spacetime whose constant $t=t_0$ level sets are cylinders $S^2 \times (-r_0,r_0)_r$. 
In the throat region, the null expansions 
$\theta_l, \theta_k$, as well as shears $\sigma_l, \sigma_k$ 
vanish identically and $R_{\mu\nu} = \frac{1}{\rho_0^2} q_{\mu\nu}$. With these structures, the statement of Prop. \ref{label:geomprop} is seen to still hold.
Thus, one would expect the QDEC to apply in this region. Testing it with the vector $t^\mu = (\partial_t)^\mu$ and integrating once in 
the $r$-direction between $-r_0$ and $r_0$ at fixed time $t_0$ gives:
\ben
2\pi \int\limits_{{\rm throat}} \langle T_{\mu\nu} \rangle \, t^\mu d\Sigma^\nu \ge r^\mu \partial_\mu S_{\rm EE} \Bigg|_{-r_0,t_0}^{+r_0,t_0}
\een
where $r^\mu = (\partial_r)^\mu$, and where $S_{\rm EE}$ is considered to be a function of the cut $C_{t,r}$ of fixed $r$ and $t$.

\medskip
\begin{proof} (of Prop. \ref{label:geomprop}) 
Consider the principal eigenfunction $\psi_1$ of $L$. It is known \cite{mars} that $\psi_1>0$ on $C$, so 
we may redefine $l \to \psi_1 l, k \to \psi_1^{-1}k$ on $C$. These vector fields are extended as affinely parameterized 
tangents to geodesics generating $H^-$ and $H^+$, respectively. We then get a new (gauge equivalent) connection 
$\beta$ on $C$ in \eqref{Ldef}, and with this new choice of $k,l$, the new principal eigenfunction is $\psi_1=1$ and we have
$\delta_l \theta_k=l^\mu \nabla_\mu \theta_k = 0$ on $C$. Likewise, since our assumptions give $W=0$ by \eqref{Wphi}, we have
$\delta_k \theta_k=k^\mu \nabla_\mu \theta_k = 0$ on $C$. Now we consider the terms in $k^\beta \delta S_{\rm EE}^{\rm div}[C]/\delta C^\beta$ in \eqref{Sdiv} one by one. 

 (i) Term proportional to $c_2$: We have $k^\beta \delta {\rm Area}[C]/\delta C^\beta = \theta_k$ on $C$. Taking a further derivative 
 along $u = l$ off of $C$ and using $l^\mu \nabla_\mu \theta_k = 0 = k^\mu \nabla_\mu \theta_k$ on $C$ as just argued gives 
  $u^\alpha \nabla_\alpha ( k^\beta \delta {\rm Area}[C]/\delta C^\beta) =0$ on $C$
 for all $u$ orthogonal to $C$.

(ii) Term proportional to $c_0^{\rm I}$: This term is topological and hence has a vanishing variation in any direction to any order off of $C$. 

(iii) Term proportional to $c_0^{\rm II}$: This term  is not in general topological though it is Weyl invariant \cite{Solodukhin:2008dh}. By the Einstein equation, 
the contribution from the Ricci-tensor produces again terms proportional to the Euler characteristic or area, treated already in (ii), (i). 
We are thus left with the integral over $\theta_l \theta_k$. For the variation \eqref{vardef}, we find
\ben
\label{auxdiv}
k^\alpha \frac{\delta}{\delta C^\alpha} \int\limits_C \theta_k \theta_l \, dA = W\theta_l + L\theta_k.
\een
The assumptions and definition of $W$ \eqref{Wphi} gives $W\theta_l=0=u^\mu \nabla_\mu (W\theta_l)$ on $C$. We have previously seen in (i) that 
$u^\mu \nabla_\mu \theta_k = 0$ on $C$, and by assumption $\theta_k=0$ on $C$. Then, since by \eqref{Ldef}, $L$ only has intrinsic 
derivatives in $C$, it follows that $L\theta_k=0=u^\mu \nabla_\mu (L\theta_k)$ on $C$. Therefore, $u^\mu \nabla_\mu$ on \eqref{auxdiv}
vanishes on $C$.
\end{proof}



\subsection{Extended supersymmetry and topological twisting}
\label{sec:susy}
It is known that in theories with extended supersymmetry, there exist certain nil-potent combinations of the supercharges of the theory related to a 
so-called ``topological twist'' \cite{Witten:1988ze}. Considering for definiteness $N=2$ supersymmetric Yang-Mills theories in $d=4$ dimensional Minkowski spacetime \cite{Witten:1988ze}, key features of 
this supercharge $Q$ are  that (i) $Q$ is fermionic, (ii) $Q^2=0$, (iii) there is a bosonic local $2$-form valued operator $J[k]$ for every constant vector field 
$k$ such that 
\ben
\label{hdual}
\epsilon_{\alpha\beta\gamma}{}^\nu T_{\mu\nu} k^\mu = 3\nabla_{[\alpha} J_{\beta \gamma]}[k] + \text{$Q$-exact}, 
\een
where the $Q$-exact term indicates an operator that can be written as the graded commutator $[Q, \, . \, ]_\pm$ and some other (fermionic) operator, and $\epsilon$ is 
the positively oriented volume form on spacetime, see footnote \ref{orientationfootnote}. The expression on the right side is the exterior differential $dJ[k]$.

Any $Q$-exact operator has a vanishing expectation value in a ``Q-closed'' state $\Phi$, i.e. a state such that $Q\Phi=0$. From now we assume that $\Phi$ is such a state. 
Let $C$ be the edge of a wedge, and let $k$ be a future pointing constant null vector field tangential to the future horizon, $H^+$, of the wedge. We now 
take the dual of the QDEC as in \eqref{hdual}, and write $\epsilon_{\alpha\beta\mu\nu} = 6l_{[\alpha} k_\beta \nu_{\mu\nu]}$, where $l$ is a second future directed constant null vector orthogonal to $C$ and tangent to the past horizon $H^-$, with $l_\mu k^\mu = -1$, and 
where $dA=\nu_{\mu\sigma}dx^\mu \wedge dx^\sigma$ is the positively oriented volume form on $C=\partial H^-$, see footnote \ref{orientationfootnote}. We make these substitutions and integrate the dualized QDEC 
over $H^-$, applying Stokes' theorem to the exact 3-form given by the expectation value of \eqref{hdual} in our $Q$-closed state $\Phi$. Using that the bi-normal to 
$C$ is $n_{\alpha\beta} = 2l_{[\alpha} k_{\beta]}$, we find
\ben
\label{susy}
2\pi \int\limits_C \langle J_{\mu\nu}[k] \rangle \, dx^\mu \wedge dx^\nu \ge \int\limits_C k^\mu \frac{\delta}{\delta C^\mu} S_{\rm EE}  \, dA . 
\een
Such a relation will give an upper bound if the expectation value $\langle J_{\mu\nu}[k] \rangle=(\Phi, J_{\mu\nu}[k] \, \Phi)$ could be computed, e.g., via localization techniques. 

Actually, by taking the graded commutator of \eqref{hdual} with the supercharge $Q$, we see that $J[k] \equiv J^{(2)}[k]$ itself is $Q$-exact up to a 
$d$-closed 1-form. It is in fact the first form of the Witten-Donaldson ladder related to the Donaldson-invariants for 4-manifolds \cite{Witten:1988ze}. In other words, 
there are $p$-forms such that 
\ben
\label{hdual1}
\left[Q, J^{(p)}[k]\right]_\pm = dJ^{(p-1)}[k], 
\een
for $p=1,2$. Thus, the left side in \eqref{susy} (formally) does not change if we replace $\Phi \to \Phi + Q\Psi$, i.e. it is an invariant of the cohomology class of $\Phi$. 
One would expect the right side also to have this property. 

\section{Summary and conclusions}

In this work, we have proposed a QDEC \eqref{qdec} for certain cuts in certain curved spacetimes. We have demonstrated this inequality 
when the cut is the edge of a Rindler wedge, in a general setting of algebraic quantum field theory. It would be interesting to prove the more general 
version of the QDEC, too. This would require a better understanding of the relation between a suitably regulated version of the entanglement entropy and the relative entropy, 
which has been the primary tool in this paper.  
It would also be interesting to understand  the topological setup considered in Sec. \ref{sec:susy} better, e.g., what is the precise form of the QDEC for topological states $\Phi$
in quantum gauge theories with extended supersymmetry. Finally, one could reconsider classical theorems in general relativity based on the DEC in view of the QDEC.

\medskip
\noindent
{\bf Acknowledgement.} We thank Aron Wall and Zihan Yan for correspondence. SH thanks Philip Argyres for discussions.
This work has been supported by Klaus Tschira Foundation.


\begin{thebibliography}{99}
\bibitem{ChatGPT}
The figure is based on an original image created by the authors. 
The corresponding TikZ code was generated with the assistance of ChatGPT (OpenAI) 
and subsequently refined by the authors.
\end{thebibliography}
\end{document}